\documentclass[12pt]{elsarticle}
\usepackage[margin=2.5cm]{geometry}
\usepackage{lipsum}

\makeatletter
\def\ps@pprintTitle{%
   \let\@oddhead\@empty
   \let\@evenhead\@empty
   \def\@oddfoot{\reset@font\hfil\thepage\hfil}
   \let\@evenfoot\@oddfoot
}
\makeatother

\usepackage{amsfonts,color,morefloats,pslatex}
\usepackage{amssymb,amsthm, amsmath,latexsym}
\usepackage{arydshln}

\newtheorem{theorem}{Theorem}
\newtheorem{lemma}[theorem]{Lemma}

\newtheorem{corollary}[theorem]{Corollary}

\newtheorem{definition}[theorem]{Definition}
\newtheorem{example}[theorem]{Example}

\newtheorem{conj}[theorem]{Conjecture}
\newtheorem{remark}[theorem]{Remark}

\newcommand{\tr}{{\mathrm{Tr}}}

\newcommand{\Norm}{{\mathrm{N}}}
\newcommand{\gf}{{\mathbb{F}}}

\newcommand{\support}{{\mathrm{suppt}}}

\newcommand{\cC}{{\mathcal{C}}}
\newcommand{\cD}{{\mathcal{D}}}
\newcommand{\cRF}{{\mathcal{RF}}}

\newcommand{\ba}{{\mathbf{a}}}

\newcommand{\bc}{{\mathbf{c}}}
\newcommand{\bg}{{\mathbf{g}}}

\usepackage{enumerate}
\usepackage{booktabs}
\usepackage{arydshln}
\allowdisplaybreaks[4]
\begin{document}

\begin{frontmatter}

\title{Two families of linear codes with desirable properties from some functions over finite fields}

\tnotetext[fn1]{
Z. Heng's research was supported in part by the National Natural Science Foundation of China under Grant 12271059 and 11901049, and in part by the Fundamental Research Funds for the Central Universities, CHD, under Grant 300102122202; Y. Wu's research was supported in part the National Natural Science Foundation of China (Grant No. 12101326) and in part by the Natural Science Foundation of Jiangsu Province (Grant No. BK20210575); Q. Wang's reasearch was supported in part by the National Natural Science Foundation of China (Grant No. 11931005, 12371522).
}

\author[1]{Ziling Heng}
\ead{zilingheng@chd.edu.cn}
\author[1]{Xiaoru Li$^{\ast}$}
\ead{lx\underline{ }lixiaoru@163.com}
\author[2]{Yansheng Wu}
\ead{yanshengwu@njupt.edu.cn}
\author[3]{Qi Wang}
\ead{wangqi@sustech.edu.cn}

\cortext[cor]{Corresponding author}
\address[1]{School of Science, Chang'an University, Xi'an 710064, China}
\address[2]{School of Computer Science, Nanjing University of Posts and Telecommunications, Nanjing 210023, China}
\address[3]{Department of Computer Science and Engineering, National Center for Applied Mathematics Shenzhen, Southern University of Science and Technology, Shenzhen, 518055, China}

\begin{abstract}
Linear  codes are widely studied in coding theory as they have nice applications in  distributed storage, combinatorics, lattices, cryptography and so on.
Constructing linear codes with desirable properties is an interesting research topic.  In this paper, based on the augmentation technique, we present two families of linear codes from some functions over finite fields. The first family of linear codes is constructed from monomial functions over finite fields. The locality of them is determined and the weight distributions of two subfamilies of the codes are also given. An infinite family of locally recoverable codes which are at least almost optimal and some optimal recoverable codes are obtained from the linear codes. In particular, the two subfamilies of the codes are proved to be both optimally or almost optimally extendable and self-orthogonal.
The second family of linear codes is constructed from weakly regular bent functions over finite fields and their weight distribution is determined. This family of codes is proved to have locality 3 for some cases and is conjectured to have locality 2 for other cases. Particularly, two families of optimal locally recoverable codes are derived from the linear codes. Besides, this family of codes is also proved to be both optimally or almost optimally extendable and self-orthogonal.
\end{abstract}

\begin{keyword}
Weakly regular bent functions \sep monomial functions \sep linear codes

\MSC  94B05 \sep 94A05

\end{keyword}

\end{frontmatter}

\section{Introduction}\label{sec1}
Let $\gf_q$ be the finite field with $q$ elements, where $q$ is a  power of a prime $p$. Denote by $\gf_q^*$ the set of all the non-zero elements in $\gf_q$.
\subsection{Linear codes and self-orthogonal codes}
For a positive integer $n$, if $\cC$ is a $k$-dimensional linear subspace of $\gf_q^n$, then it is called an $[n,k,d]$ linear code over $\mathbb{F}_q$, where $d$ denotes its minimum  distance. Let $A_{i}$ denote the number of codewords with weight $i$ in $\cC$ of length $n$, where $0 \leq i \leq n$. The sequence $(1,A_{1},A_{2}, \cdots ,A_{n})$ is called the weight distribution of $\cC$. The weight enumerator of $\cC$ is defined by
$
A(z)=1+A_{1}z+A_{2}z^2+ \cdots +A_{n}z^n.
$
Weight distribution is an important research subject as it not only describes the error detection and error correction ability of the code, but also can be used to calculate the error probability of error detection and correction of the code. The weight distributions of linear codes have been extensively studied in the literature \cite{D1, D2, D3, D4, HWL, HZ4, LC, HZ3, ZD, Zhou}. Optimal linear code plays an important role in both theory and practice. An $[n, k, d]$ linear code $\cC$ is said to be optimal if no $[n, k, d + 1]$ code exists and almost optimal if there exists an $[n, k, d + 1]$ optimal code.
There exist some tradeoffs among the parameters of linear codes. The followings are the well known sphere packing bound and Plotkin bound, where $\lfloor \cdot\rfloor$ denotes the floor function.
\begin{lemma}\cite[Sphere packing bound]{H}\label{sphere}
  Let $M$ be the maximum number of codewords in a code over $\gf_q$ of length $n$ and minimum distance $d$. Then
  \begin{eqnarray*}
    M \leq \frac{q^n}{\sum\limits_{i=0}^{t}\tbinom{n}{i}(q-1)^i},
  \end{eqnarray*}
  where $t=\lfloor(d-1)/2\rfloor$.
\end{lemma}

\begin{lemma}\cite[Plotkin bound]{H}
  Let $M$ be the maximum number of codewords in a code over $\gf_q$ of length $n$ and minimum distance $d$. If $(1-q^{-1})n < d$, then
  \begin{eqnarray*}
    M \leq \left\lfloor \frac{d}{d-(1-q^{-1})n} \right\rfloor.
  \end{eqnarray*}
\end{lemma}

Define the dual code of an $[n,k]$ linear code $\cC$ as
$$
\mathcal{C}^{\perp}=\left\{ \mathbf{u} \in \mathbb{F}_{q}^{n}: \langle  \mathbf{u}, \mathbf{v} \rangle=0 \mbox{ for all }\mathbf{v} \in \mathcal{C} \right\}
,$$
where $\langle \cdot \rangle$ denotes the standard inner product. Obviously, $\cC^{\perp}$ is an $[n,n-k]$ linear code.
If a linear code $\cC$ satisfies $\cC \subseteq \cC^{\perp}$, then it is called a self-orthogonal code. If a linear code $\cC$ satisfies $\cC = \cC^{\perp}$, then it is said to be self-dual. For a linear code $\cC$ over $\gf_q$, if all codewords in $\cC$ have weights divisible by an integer $\Delta>1$, then $\cC$ is said to be divisible \cite{H}. Furthermore, $\cC$ is said to be $\Delta$-divisible and $\Delta$ is called a divisor of $\cC$ \cite{KK}.

For a cyclic code $\cC$ with generator polynomial $g(x)$,  $\cC$ is self-orthogonal if and only if $g^{\perp}(x) \mid g(x)$, where
$g^{\perp}(x)=x^kh(x^{-1})/h(0)$ with $h(x):=\frac{x^n-1}{g(x)}$  \cite{H}.
For binary and ternary linear codes, there exist simple  conditions for them to be self-orthogonal by the divisibility of their weights \cite{H}.
For general $q$-ary linear codes, however, it is very difficult to judge whether they are self-orthogonal or not.
Sometimes we may directly prove a linear code to be self-orthogonal by definition, i.e. $\cC$ is self-orthogonal if and only if $\bc \cdot \bc' = 0$ for any codewords $\bc, \bc' \in \cC$. Recently, in \cite{HZ5}, Li and Heng  established a sufficient condition for a  $q$-ary linear code containing all-$1$ vector to be self-orthogonal if $q$ is a power of an odd prime.
\begin{lemma}\cite{HZ5}\label{lem-self-orthogonal}
  Let $q=p^e$, where $p$ is an odd prime. Let $\cC$ be an $[n,k,d]$ linear code over $\gf_q$ with $\mathbf{1} \in \cC$, where $\mathbf{1}$ is the all-$1$ vector of length $n$. If $\cC$ is $p$-divisible, then $\cC$ is self-orthogonal.
\end{lemma}

Self-orthogonal codes have nice applications in many areas, such as quantum codes, lattices and LCD codes. In \cite{Z}, several infinite families of binary LCD codes and self-orthogonal codes were obtained from a generic construction.
In \cite{HZL}, several families of ternary self-orthogonal codes were constructed from weakly regular bent functions and new infinite families of ternary LCD codes were derived by using these self-orthogonal codes. In \cite{WH}, several families of $q$-ary linear codes were constructed by some special defining sets and these codes were proved to be self-orthogonal. Besides, several families of optimal quantum codes and new LCD codes were derived from the self-orthogonal codes in \cite{WH}. In this paper, we will construct more self-orthogonal codes.
\subsection{Optimally extendable codes}
With the development of communication technologies, the implementation of cryptographic algorithms has become very important due to the increasing demand for secure transmission of confidential information. However, side channel attack (SCA) and fault injection attack (FIA) have brought great threats to the implementation of block cipher. Recently, Bringer et al. proposed a direct sum masking (DSM) countermeasure to safeguard against SCA and FIA in \cite{BJ}. This method uses two linear codes $\cC$ and $\cD$ such that $\cC \bigoplus \cD = \gf_q^n$ (i.e., $\cC$ and $\cD$ are supplementary). The codes $\cC$ and $\cD$ are used to encode the sensitive data and encode random data, respectively. The security parameters $(d(\cC)-1,d(\cD^{\perp})-1)$ quantify the ability of codes $(\cC, \cD)$ to resist SCA and FIA, where $d(\cC)-1$ is the number of intrusion errors that can be always detected, and $d(\cD^{\perp})-1$ gives the probing security order of the counter-measure \cite{BJ}. In order to protect the sensitive data stored in the register from SCA and FIA, and to protect the whole algorithm, the lengths of the linear codes $\cC$ and $\cD$ are extended by $k$ bits. By adding $0$ to the end of each codeword in linear code $\cC$, the resulting linear code $\cC'$ has the same minimum distance as $\cC$. Adding the $k \times k$ identity matrix to the right of the generator matrix of $\cD$ gives the generator matrix of the linear code $\cD'$, and the resulting $\cD'$ is an $[n+k,k]$ linear code. In this time, the security parameters change into $(d(\cC')-1, d(\cD'^{\perp})-1)=(d(\cC)-1, d(\cD'^{\perp})-1)$. It is clear that $d(\cD'^{\perp}) \leq d(\cD^{\perp})$. Hence, there is a risk that security parameters may deteriorate. It is desirable to construct linear code $\cD$ such that $d(\cD'^{\perp})$ and $d(\cD^{\perp})$ are very close.  If $d(\cD'^{\perp})=d(\cD^{\perp})$, then $\cD$ is called an \emph{optimally extendable} linear code. If $d(\cD'^{\perp})=d(\cD^{\perp})-1$, then $\cD$ is called an \emph{almost optimally extendable} linear code.

Recently,  some (almost) optimally extendable codes were constructed in the literature. Carlet et al. gave a construction of optimally extendable linear code by using algebraic geometry linear codes in \cite{Carlet2018}. To the subsequent, Carlet et al.  presented two constructions of optimally extendable linear codes using primitive irreducible cyclic codes \cite{Carlet3}. Moreover, they also proposed two families of almost optimally extendable linear codes from the first-order Reed-Muller codes and irreducible cyclic codes of dimension two. Recently, Quan et al. constructed three families of (almost) optimally extendable linear codes from irreducible cyclic codes, MDS codes and NMDS codes in \cite{QQ}.

Let $G$ be a generator matrix of the linear code $\cC$. Let $\cC_1$ and $\cC_2$ be the linear codes generated by matrixes $[I:G]$ and $[G:I]$, respectively. Then it is easy to deduce that linear codes $\cC_1$ and $\cC_2$ are permutation equivalent. Besides, it should be remarked that the linear code $\cC_1$ generated by $[I:G]$  depends on the choice of the generator matrix $G$ of $\cC$. Different $G$'s may yield inequivalent linear codes. In this paper, we will construct two families of (almost) optimally extendable linear codes.

\subsection{Locally recoverable codes}
In order to recover the data in distributed and cloud storage systems,
locally recoverable codes (LRCs for short) were proposed by Gopalan, Huang, Simitci and Yikhanin \cite{GH}.
For a positive integer $n$, we use the notation $[n]=\{ 0,1,\cdots,n-1 \}$.  Let $\mathcal{C}$ be an $[n,k,d]$ linear code over $\mathbb{F}_q$. We index the coordinates of the codewords in $\mathcal{C}$ with the elements in $[n]$.
For each $i \in [n]$, if there exists a subset $R_{i} \subseteq [n] \backslash {i}$ of size $r$ and a function $f_{i}(x_1,x_2,\cdots,x_r)$ on $\mathbb{F}_q^{r}$ meeting $c_i=f_{i}(\mathbf{c}_{R_i})$ for any $\mathbf{c}=(c_0,\cdots,c_{n-1}) \in \mathcal{C}$, then we call $\mathcal{C}$ an $(n,k,d,q;r)$-LRC, where $\mathbf{c}_{R_i}$ is the projection of $\mathbf{c}$ at $R_{i}$. The set $R_{i}$ is called  the repair set of $c_i$ and $r$ is referred to as the locality of $\mathcal{C}$. Locally recoverable codes have been implemented in practice by Microsoft and Facebook \cite{HC, SM}.

Since LRCs are needed in large-scale distributed storage systems, constructing LRCs with small locality has been an interesting research topic.
A lot of progress has been made in the research of locally recoverable codes. The reader may refer to \cite{CB1, CB2, FQ, FQ2, HY, LRH, HZ3, TI,  TP, LX} and their references. Very recently, in \cite{TP}, Tan et al. developed some general theory on the minimum locality of linear codes and investigated the minimum locality of some families of linear codes.

 To state a lemma by Tan et al. \cite{TP}, we need to review the definition of $t$-designs. Let $n, \kappa$ and $t$ be positive integers with $1 \leq t \leq \kappa \leq n$. Let $\mathcal{P}$ denote a set of $n$ elements and $\mathcal{B}$ denote a set of $\kappa$-subsets of $\mathcal{P}$. If every $t$-subset of $\mathcal{P}$ is contained in exactly $\lambda$ elements of $\mathcal{B}_\kappa$, then the pair $(\mathcal{P},\mathcal{B}_\kappa)$ is called a $t$-$(n,\kappa,\lambda)$ design. Linear codes are often used to construct $t$-designs (e.g., see \cite{D, D1, DX, LY, TD, TDX}).
In the following, we introduce the well-known coding-theoretic construction of $t$-designs.
Let $\mathcal{P}=\{ 0,1,\cdots,n-1 \}$ denote the set of coordinate positions of the codewords of linear code $\mathcal{C}$ with length $n$. We define the support of a codeword $\textbf{c}=\{ c_0,c_1,\cdots,c_{n-1} \}$ in $\mathcal{C}$ by $\support(\textbf{c})=\{0 \leq i \leq n-1 : c_i \neq 0 \}$. Denote by $\mathcal{B}_\kappa$ the set of supports of all codewords with Hamming weight $\kappa$ in $\mathcal{C}$. The pair $(\mathcal{P},\mathcal{B}_\kappa)$ may be a $t$-$(n,\kappa,\lambda)$ design for some positive integer $\lambda$. Then we say that the code $\mathcal{C}$ supports a $t$-$(n,\kappa,\lambda)$ design.
By the following lemma, we can derive the locality of a linear code if its dual holds a $1$-design.

\begin{lemma}\label{lem-locality}\cite{TP}
Let $\cC$ be a nontrivial linear code of length $n$ and $d^{\perp}$ be the minimum distance of $\cC^{\perp}$. If $(\mathcal{P}(\cC^{\perp}), \mathcal{B}_{d^{\perp}}(\cC^{\perp}))$ is a $1$-$(n, d^{\perp}, \lambda^{\perp})$ design with $\lambda^{\perp} \geq 1$, then $\cC$ has locality $d^{\perp}-1$.
\end{lemma}

Similarly to the classical linear codes, there exist some tradeoffs among the parameters of LRCs.
We list two famous bounds on LRCs in the following.

\begin{lemma}[\cite{GH}, Singleton-like bound]\label{lem-Slbound}
For any $(n,k,d,q;r)$-LRC,
\begin{eqnarray}\label{eqn-Slbound}
d \leq n-k- \left \lceil \frac{k}{r} \right \rceil +2.
\end{eqnarray}
\end{lemma}

LRCs are said to be distance-optimal ($d$-optimal for short) when they achieve  the Singleton-like bound. LRCs are said to be almost distance-optimal (almost $d$-optimal for short) when they meet the Singleton-like bound minus one with equality.

\begin{lemma}[\cite{Upper}, Cadambe-Mazumdar bound]\label{lem-CMbound}
For any $(n,k,d,q;r)$-LRC,
\begin{eqnarray}\label{eqn-CMbound}
k \leq \mathop{\min}_{t \in \mathbb{Z}^{+}} [rt+k_{opt}^{(q)}(n-t(r+1),d)],
\end{eqnarray}
where $k_{opt}^{(q)}(n,d)$ is the largest possible dimension of a linear code with length $n$, minimum distance $d$ and alphabet size $q$, and $\mathbb{Z}^{+}$ denotes the set of all positive integers.
\end{lemma}

LRCs are said to be dimension-optimal ($k$-optimal for short) when they achieve the Cadambe-Mazumdar bound. LRCs are said to be almost dimension-optimal (almost $k$-optimal for short) when they meet the Cadambe-Mazumdar bound minus one with equality. A locally recoverable code that achieves either of these two bounds is called optimal.

\subsection{The objectives of this paper}
As linear  codes have nice applications in  distributed storage, combinatorics, cryptography and so on,
constructing linear codes with desirable properties is an interesting research topic in coding theory.
In order to illustrate the main purposes of this paper, we first recall a well-known method for constructing linear codes, i.e., the defining set method. Let $D=\{d_1, d_2, \cdots, d_n\} \subseteq \gf_{q^m}$. Define the trace function from $\gf_{q^m}$ to $\gf_q$ by $\tr_{q^m/q}(x)=x+x^q+x^{q^2}+\cdots+x^{q^{m-1}},\ x \in \gf_{q^m}.$ Define
\begin{eqnarray*}\label{eq-CD}
\cC_D=\left\{\left(\tr_{q^m/q}(bd_1), \tr_{q^m/q}(bd_2), \cdots, \tr_{q^m/q}(bd_n)\right):b\in \gf_{q^m}\right\}
\end{eqnarray*}
  which is a linear code of length $n$ over $\gf_q$. The set $D$ is called the defining set of $\cC_D$. In recent years, a large number of good linear codes have been constructed by special defining sets \cite{D1, D2, D3, D4, HZ4, LC}.
  The augmented code of $\cC_D$ is defined by
\begin{eqnarray}\label{eq-CD1}
\overline{\cC_{D}}=\{(\tr_{q^m/q}(bx))_{x \in D} + c\mathbf{1}:b \in \gf_{q^m}, c \in \gf_q\},
\end{eqnarray}
where $\mathbf{1}$ is the all-1 vector of length $|D|$. If $\mathbf{1}\not\in \cC_D$, then $\overline{\cC_{D}}$ have larger dimension than $\cC_D$. This means that the code rate of $\overline{\cC_{D}}$ is bigger than that of $\cC_D$. As was pointed out by Ding and Tang in \cite{D}, we may require information of the complete weight distribution of the original code $\cC_D$ to determine the minimal distance and weight distribution of the augmented code $\overline{\cC_{D}}$. However, it is in general very difficult to determine the complete weight distribution of a linear code. Hence, it is usually a challenge to determine the augmented code of a linear code.

The first objective of this paper is to study the augmented code $\overline{\cC_D}$  in (\ref{eq-CD1}) with the defining set
\begin{eqnarray}\label{eq-D}
D=\{x \in \gf_{q^m} : \tr_{q^t/q}(x^N)=0\},
\end{eqnarray}
where $N$ is a positive integer and $t$ is the least positive integer such that $\frac{q^m-1}{\text{gcd}(q^m-1, N)} \mid (q^t-1)$.
Note that $D$ is constructed from the monomial function $g(x)=x^N$ over $\gf_{q^m}$.
Firstly, we list some known results on the code $\cC_{D\setminus \{0\}}$ in the following:
\begin{enumerate}[l]
  \item[$\bullet$] If $m=2r$, $N=q^r+1$ and $q$ is an odd prime, the code $\cC_{D\setminus \{0\}}$ is a two-weight code and its complete weight enumerator was given in \cite{LC}.
  \item[$\bullet$] If $m=2r$, $N=q^r+1$ and $q$ is an any prime power, then the code $\cC_{D\setminus \{0\}}$ is a two-weight code and its weight enumerator was given in \cite{HZ4}.
  \item[$\bullet$] If $N=2$ and $q$ is an odd prime, then the weight distribution of code $\cC_{D\setminus \{0\}}$ was given in \cite{D4}.
\end{enumerate}
In Section \ref{sec4}, we  will prove that the locality of $\overline{\cC_D}$ is $2$ for general positive integer $N$  if $q>2$.
Some optimal or almost optimal LRCs are obtained.  If $q=2$, an infinite family of almost $k$-optimal LRCs is derived.
For $N=2$ and odd prime power $q$ or $N=q^r+1$ and any prime power $q$,  $\overline{\cC_D}$ is proved to be self-orthogonal and its weight distribution is determined.
From the weight distribution of $\overline{\cC_D}$, we find that the minimum distance of $\overline{\cC_D}$ is the same with that of $\cC_{D\setminus \{0\}}$. However, the dimension and the code rate of $\overline{\cC_D}$ are both larger than those of $\cC_{D\setminus \{0\}}$.
If $q=2$ and $N=q^r+1$, it is also proved that $\overline{\cC_D}$ holds $2$-designs.
Besides, if $N=2$ or $N=q^r+1$, then $\overline{\cC_D}$ is proved to be an optimally or almost optimally extendable code by choosing suitable generator matrix.

The second objective of this paper is to construct another family of linear codes from weakly regular bent functions over finite fields.
Let $p$ be an odd prime and $f(x)$ be a weakly regular bent function from $\gf_q$ to $\gf_p$ with $f(0)=0$. We
define a $p$-ary linear code as
\begin{eqnarray}\label{eq-Cfbar}
\cC_f=\left\{\bc_{(a,b)}=(af(x)+\tr_{q/p}(bx)+c)_{x \in \gf_q}: a \in \gf_p, b \in \gf_q, c \in \gf_p\right\}.
 \end{eqnarray}
We remark that $\cC_f$ is the augmented code of extended code of
\begin{eqnarray*}
  \cC_f^* =\left\{\bc_{(a,b)}=(af(x)-\tr_{q/p}(bx))_{x \in \gf_q^*}: a \in \gf_p, b \in \gf_q\right\},
\end{eqnarray*}
which was studied in \cite{MS}. In Section \ref{sec5}, the parameters and weight distribution of $\cC_f$ are determined. This family of linear codes contains a subfamily holding $2$-designs. From the weight distribution of $\cC_f$, we note that the minimum distance of $\cC_f$ is the same with that of $\cC_f^*$. However, the dimension and the code rate of $\cC_f$ are both larger than those of $\cC_f^*$. Besides, $\cC_f$ is also proved to be self-orthogonal.
By choosing suitable generator matrix of $\cC_f$, $\cC_f$ is proved to be optimally or almost optimally extendable.
Furthermore, $\cC_f$ is proved to have locality $3$ for some cases and is conjectured to have locality $2$ for other cases. Two families of optimal LRCs are derived from this family of linear codes.

\section{Preliminaries}\label{sec2}
In this section, we will introduce some known results on characters, Gaussian sums over finite fields, weakly regular bent functions and cyclotomic fields.
\subsection{Characters over finite fields}
Let $q=p^e$ with $p$ a prime. Denote by $\zeta_p$ the primitive $p$-th root of complex unity.
 Define the additive character of $\gf_q$ by the homomorphism from the additive group $\gf_q$ to the complex multiplicative group $\mathbb{C}^*$ such that
$
\phi(x+y)=\phi(x)\phi(y)
$ for all $x, y \in \gf_q$.
For any $a \in \gf_q$,
an additive character of $\gf_q$ can be defined by the function $\phi_a(x)=\zeta_p^{\tr_{q/p}(ax)},\ x\in \gf_q$, where $\tr_{q/p}(x)$ is the trace function from $\gf_q$ to $\gf_p$. In addition, the set $\widehat{\gf_q}=\{\phi_a:a\in \gf_q\}$ gives all $q$ different additive characters of $\gf_q$. By definition, we have $\phi_a(x)=\phi_1(ax)$. In particular, $\phi_0$ is referred to as the trivial additive character of $\gf_q$ and $\phi_1$ is called the canonical additive character of $\gf_q$. The orthogonal relation of additive characters \cite{L} is given by
\begin{eqnarray*}
\sum_{x\in \gf_q}\phi_a(x)=\begin{cases}
q    &\text{if $a=0,$ }\\
0     &\text{otherwise .}
\end{cases}
\end{eqnarray*}

Let $\gf_q^*=\langle\beta\rangle$. For each $0\leq j\leq q-2$, a multiplicative character $\psi(x)$ of $\gf_q$ is defined as the homomorphism from the multiplicative group $\gf_q^*$ to the complex multiplicative group $\mathbb{C}^*$ such that
$
\psi(xy)=\psi(x)\psi(y)
$ for all $x, y \in \gf_q^*$.
The function
$\psi_j(\beta^k)=\zeta_{q-1}^{jk}$ for $k=0,1,\cdots,q-2$
 gives a multiplicative character, where $0\leq j \leq q-2$. The set $\widehat{\gf_q^*}=\{\psi_j:j=0,1,\cdots,q-2\}$ gives all multiplicative characters of $\gf_q$ and  is a multiplicative group of order $q-1$. In particular, $\psi_0$ is called the trivial multiplicative character and $\eta:=\psi_{\frac{q-1}{2}}$ is referred to as the quadratic multiplicative character of $\gf_q$ if $q$ is odd. The orthogonal relation of multiplicative characters (see \cite{L}) is given by
\begin{eqnarray*}
\sum_{x\in\gf_q^*}\psi_j(x)=\begin{cases}
q-1    &\text{if $j=0,$ }\\
0    &\text{if $j \neq 0$.}
\end{cases}
\end{eqnarray*}

\subsection{Gaussian sums over finite fields}
For an additive character $\phi$ and a multiplicative character $\psi$ of $\gf_q$, the \emph{Gaussian sum} $G(\psi, \phi)$ over $\gf_q$ is defined by
$$G(\psi, \phi)=\sum_{x\in \gf_q^*}\psi(x)\phi(x).$$
Specially, $G(\eta,\phi)$ is referred to as the \emph{quadratic} \emph{Gaussian sum} over $\gf_q$ for nontrivial $\phi$.

The explicit values of quadratic \emph{Gaussian sum} are given as follows.
\begin{lemma}[\cite{L}, Theorem 5.15]\label{quadGuasssum1}
Let $q=p^e$ with $p$ odd. Let $\phi_1$ be a canonical additive character of $\gf_q$ and $\eta$ be a quadratic multiplicative character of $\gf_q$. Then
\begin{eqnarray*}
G(\eta,\phi_1)=(-1)^{e-1}(\sqrt{-1})^{(\frac{p-1}{2})^2e}\sqrt{q}
=\left\{
\begin{array}{lll}
(-1)^{e-1}\sqrt{q}    &   \mbox{ for }p\equiv 1\pmod{4},\\
(-1)^{e-1}(\sqrt{-1})^{e}\sqrt{q}    &   \mbox{ for }p\equiv 3\pmod{4}.
\end{array}
\right.
\end{eqnarray*}
\end{lemma}

The following lemmas give the explicit values of some special exponential sums.
\begin{lemma}[\cite{L}, Theorem 5.33]\label{lem-weil}
Let $\phi$ be a nontrivial additive character of $\gf_q$, where $q$ is power of an odd prime. Let $f(x)=a_2x^2+a_1x+a_0 \in \gf_q[x]$, where $a_2 \neq 0$. Then
\begin{eqnarray*}
\sum_{c \in \gf_q}\phi\left(f(c)\right) = \phi\left(a_0-a_1^2(4a_2)^{-1}\right)\eta(a_2)G(\eta, \phi).
\end{eqnarray*}
\end{lemma}

\begin{lemma}\label{lem-sums}\cite{L}
  Let $\gf_q$ be a finite field with characteristic $p$ and let
  \begin{eqnarray*}
    f(x) &=& a_rx^{p^r}+a_{r-1}x^{p^{r-1}}+\cdots+a_1x^p+a_0x+a
  \end{eqnarray*}
  be an affine $p$-polynomial over $\gf_q$. Let $\phi_b$, $b \in \gf_q^*$, be a nontrivial additive character of $\gf_q$. Then
  \begin{eqnarray*}
    \sum_{c \in \gf_q}\phi_b(f(c)) &=& \begin{cases}
                                         \phi_b(a)q & \mbox{if } ba_r+b^pa_{r-1}^{p} + \cdots +b^{p^{r-1}}a_1^{p^{r-1}}+b^{p^r}a_0^{p^r}=0, \\
                                         0 & \mbox{otherwise}.
                                       \end{cases}
  \end{eqnarray*}
\end{lemma}

The following lemma can be obtained directly from  \cite[Lemma $7$]{HZ2}.
\begin{lemma}\label{lem-2k}
Let $m=2r$ be an even integer. Then we have
\begin{eqnarray*}
\sum_{x \in \gf_{q^{m}}}\zeta_p^{\tr_{q^r/p}(ax^{q^r+1})+\tr_{q^{m}/p}(bx)}=
-q^r\zeta_p^{-\tr_{q^r/p}(\frac{b^{q^r+1}}{a})},
\end{eqnarray*}
where $a \in \gf_{q^r}^*$, $b \in \gf_{q^{m}}$ and $\zeta_p$ is the complex $p$-th root of unity.
\end{lemma}

\subsection{The number of roots of quadratic polynomial over finite fields}
The number of the roots of the quadratic polynomial over finite field of characteristic $2$ is given in the following lemma.
\begin{lemma}\label{lem-root}\cite{BE}
  Let $\gf_q$ be a finite field of characteristic $2$ and $f(x)=ax^2+bx+c \in \gf_q[x]$ be a polynomial of degree $2$. Then
  \begin{enumerate}[1.]
    \item[\textcircled{$1$}] $f(x)$ has exactly one root in $\gf_q$ if and only if $b=0$;
    \item[\textcircled{$2$}] $f(x)$ has exactly two root in $\gf_q$ if and only if $b \neq 0$ and $\tr_{q/2}(\frac{ac}{b^2})=0$;
    \item[\textcircled{$3$}] $f(x)$ has no root in $\gf_q$ if and only if $b \neq 0$ and $\tr_{q/2}(\frac{ac}{b^2})=1$.
  \end{enumerate}
\end{lemma}

\subsection{Weakly regular bent functions and the cyclotomic fields}
Let $f(x)$ be a function from $\gf_{p^e}$ to $\gf_p$. Define the Walsh transform of $f(x)$ as follows:
\begin{eqnarray*}
\text{W}_f(\beta):=\sum_{x \in \gf_{p^e}}\zeta_p^{f(x)-\tr_{p^e/p}(\beta x)},\ \beta \in \gf_{p^e}.
\end{eqnarray*}
The function $f(x)$ is referred to as a $p$-ary bent function if $|\text{W}_f(\beta)|=p^{\frac{e}{2}}$ for any $\beta \in \gf_{p^e}$. For a bent function $f(x)$, if there exists some $p$-ary function $f^*(x)$ such that $\text{W}_f(\beta)=p^{\frac{e}{2}}\zeta_p^{f^*(\beta)}$ for any $\beta \in \gf_{p^e}$, then $f(x)$ is called a regular bent function. If there exists some $p$-ary function $f^*(x)$ and a complex $u$ with unit magnitude satisfying $\text{W}_f(\beta)=up^{\frac{e}{2}}\zeta_p^{f^*(\beta)}$ for any $\beta \in \gf_{p^e}$, then $f(x)$ is called a weakly regular bent function, where $f^*(x)$ is called the dual of $f(x)$. By \cite{Helleseth1} and \cite{Helleseth2}, if $f(x)$ is a weakly regular bent function, then
\begin{eqnarray}\label{eq-Wf}
\text{W}_f(\beta)=\varepsilon \sqrt{p^*}^e\zeta_p^{f^*(\beta)},
\end{eqnarray}
where $\varepsilon=\pm1$ is known as the sign of the Walsh transform of $f(x)$ and $p^*=(-1)^{\frac{p-1}{2}}p$. Note that the dual of a weakly regular bent function $f(x)$ is also a weakly regular bent function and $(f^*)^*(x)=f(-x)$. The sign of the Walsh transform of $f^*(x)$ is $\eta_0^e(-1)\varepsilon$.
\begin{definition}
Let $\mathcal{RF}$ denote the set of all $p$-ary weakly regular bent functions $f(x)$ satisfying $f(0)=0$ and $f(ax)=a^hf(x)$ for any $x \in \gf_{p^e}$ and $a \in \gf_p^*$, where $h$ is a positive even integer with $\gcd(h-1, p-1)=1$.
\end{definition}

Note that almost all known weakly regular bent functions are contained in $\mathcal{RF}$. The readers are referred to \cite[Table 1]{MS} for known weakly regular bent functions.

\begin{lemma}\cite{T}\label{lem-f*1}
  If $f(x) \in \cRF$, then $f^*(0)=0$.
\end{lemma}
\begin{lemma}\cite{T}\label{lem-f*2}
  If $f(x) \in \cRF$, then for any $a \in \gf_p^*$ and $x \in \gf_{p^e}$, there exists a positive even integer $l$ with $\gcd(l-1, p-1)=1$ such that $f^*(ax)=a^lf^*(x)$.
\end{lemma}

Lemmas \ref{lem-f*1} and \ref{lem-f*2} indicate that  $f^*(x) \in \cRF$ if $f(x) \in \cRF$.

The following lemma gives some results on the cyclotomic field $\mathbb{Q}(\zeta_p)$.
\begin{lemma}\cite{I}\label{lem-cyclo}
Let $K=\mathbb{Q(}\zeta_p)$ denote the $p$-th cyclotomic field over the rational number field $\mathbb{Q}$. Then the followings hold.
\begin{itemize}
  \item The ring of integers in $K$ is $O_K=\mathbb{Z}(\zeta_p)$ and $\{\zeta_p^i: 1 \leq i \leq p-1\}$ is an integral basis of $O_K$, where $\zeta_p$ is the primitive $p$-th root of complex unity.
  \item The field extension $K/\mathbb{Q}$ is Galois of degree $p-1$ and Galois group $\text{Gal}(K/\mathbb{Q})=\{\sigma_a: a \in \gf_p^*\},$
  where the automorphism $\sigma_a$ of $K$ is defined by $\sigma_a(\zeta_p)=\zeta_p^a$.
  \item The field $K$ has a unique quadratic subfield $L=\mathbb{Q}(\sqrt{p^*})$. For $1 \leq a \leq p-1$, $\sigma_a(\sqrt{p^*})=\eta_0(a)\sqrt{p^*}$, where $p^*=(-1)^{\frac{p-1}{2}}p$ and $\eta_0$ is the quadratic multiplicative character of $\gf_p$. Hence, the Galois group $\text{Gal}(L/\mathbb{Q})=\{1, \sigma_\gamma\}$, where $\gamma$ is a nonsquare in $\gf_p^*$.
\end{itemize}
\end{lemma}
According to Lemma \ref{lem-cyclo}, we have
\begin{eqnarray}\label{eq-sigma}
\sigma_a(\zeta_p^b)=\zeta_p^{ab} \mbox{ and } \sigma_a(\sqrt{p^*}^e)=\eta_0^e(a)\sqrt{p^*}^e.
\end{eqnarray}

\section{The first family of linear codes $\overline{\cC_D}$}\label{sec4}
In this section, let $\overline{\cC_D}$ be the augmented code defined in Equation (\ref{eq-CD1}) with the defining set $D=\{x \in \gf_{q^m}: \tr_{q^t/q}(x^N)=0\}$, where
$t$ is the least positive integer such that $\frac{q^m-1}{\text{gcd}(q^m-1, N)} \mid (q^t-1)$.
We will first determine the dual distance and the locality of $\overline{\cC_D}$. Then we study  $\overline{\cC_D}$ for $N=2$ and $N=q^r+1$, respectively.

From now on, denote by $\chi_1$, $\phi_1$ and $\lambda_1$ the canonical additive characters of $\gf_{q^m}$, $\gf_q$ and $\gf_p$, respectively. Let $\eta'$, $\eta$ and $\eta_0$ respectively denote the quadratic multiplicative characters of $\gf_{q^m}$, $\gf_q$ and $\gf_p$.

\subsection{The dual distance and the locality of $\overline{\cC_D}$}
In this subsection, we will determine the dual distance and the locality of $\overline{\cC_D}$.
The conventional definition of linear locally recoverable codes is presented  as follows.

\begin{definition}\label{Def-locality}\cite{HY}
Let $\cC$ be a linear code over $\gf_q$ with a generator matrix $G=[\bg_1, \bg_2, \cdots, \bg_n]$. If $\bg_i$ is a linear combination of $l (\leq r)$ other columns of $G$, then
$r$ is called the locality of the $i$-th symbol of each codeword of $\cC$. Besides, $\cC$ is called a locally recoverable code with locality $r$ if all the symbols of codeword of $\cC$ have locality $r$.
\end{definition}

In the following, we will use Definition \ref{Def-locality} to determine the locality of $\overline{\cC_D}$.

\begin{theorem}\label{loc}
Let $q >2$ be a prime power and $m \geq 2$ be an integer. Then  $\overline{\cC_D}$ is a locally recoverable code with locality $2$.
Moreover, the dual of the augmented code $\overline{\cC_D}$ has minimum distance $3$.
\end{theorem}
\begin{proof}
Let $\gf_{q^m}^*=\langle\alpha\rangle$. Then $\{\alpha^0, \alpha^1, \cdots, \alpha^{m-1}\}$ is a $\gf_q$-basis of $\gf_{q^m}$. Let $d_1, d_2, \cdots, d_{n-1},d_n$ be all the elements in $D$. For convenience, let  $d_n = 0$ due to $0\in \cC_D$.
By definition, the generator matrix $G$ of $\overline{\cC_D}$ is given by
\begin{eqnarray*}
G:=\left[
\begin{array}{cccc}
1&1&\cdots&1 \\
 \tr_{q^m/q}(\alpha^{0}d_{1})& \tr_{q^m/q}(\alpha^{0}d_{2})& \cdots &\tr_{q^m/q}(\alpha^{0}d_{n}) \\
\tr_{q^m/q}(\alpha^{1}d_{1})& \tr_{q^m/q}(\alpha^{1}d_{2})& \cdots &\tr_{q^m/q}(\alpha^{1}d_{n}) \\
\vdots &\vdots &\ddots &\vdots \\
\tr_{q^m/q}(\alpha^{m-1}d_{1})& \tr_{q^m/q}(\alpha^{m-1}d_{2})& \cdots &\tr_{q^m/q}(\alpha^{m-1}d_{n})
\end{array}\right].
\end{eqnarray*}

In the following, we first determine the locality of $\overline{\cC_D}$.
For convenience, we assume that $\bg_i$ denotes the $i$-th column of $G$, i.e., $\bg_i=(1, \tr_{q^m/q}(\alpha^0d_i), \tr_{q^m/q}(\alpha^1 d_i), \cdots, \tr_{q^m/q}(\alpha^{m-1}d_i))^T$, where $1 \leq i \leq n$.
 Note that $kd_i\in D$ for any $k\in \gf_q^*$ if $d_i\in D$. For fixed $\bg_i$ $(1 \leq i \leq n-1)$ and any $u\in \gf_q\backslash\{0,1\}$, then there exists $v=1-u$ such that $d_j := u^{-1}d_i\in D$ and
 \begin{eqnarray*}
\left\{\begin{array}{ccc}
1&=&u+v,\\
  \tr_{q^m/q}(\alpha^{0}d_{i})&=&u\tr_{q^m/q}(\alpha^{0}d_j)+v\tr_{q^m/q}(\alpha^{0}d_n), \\
  \tr_{q^m/q}(\alpha^{1}d_{i})&=&u\tr_{q^m/q}(\alpha^{1}d_j)+v\tr_{q^m/q}(\alpha^{1}d_n), \\
 & \vdots &\\
  \tr_{q^m/q}(\alpha^{m-1}d_{i})&=&u\tr_{q^m/q}(\alpha^{m-1}d_j)+v\tr_{q^m/q}(\alpha^{m-1}d_n),
\end{array}\right.
\end{eqnarray*}
where $d_n=0$. This means that $\bg_i$ is a linear combination of $\bg_j$ and $\bg_n$, where $1 \leq i \leq n-1$.
Since $v\neq 0,1$, it is clear that $\bg_n$ is also a linear combination of $\bg_i$ and $\bg_j$.
Then $\overline{\cC_D}$ is a locally recoverable code with locality $2$ according to Definition \ref{Def-locality}.

Then we determine the dual distance of $\overline{\cC_D}$.
Let $d^{\perp}$ denote the minimum distance of $\overline{\cC_D}^{\perp}$. Assume that $\bg_i$ and $\bg_j$ are any two columns of $G$. Then $\bg_i$ and $\bg_j$ are $\gf_q$-linearly dependent if and only if
\begin{eqnarray}\label{system2}
\left\{\begin{array}{c}
  \tr_{q^m/q}(\alpha^{0}d_{i})=\tr_{q^m/q}(\alpha^{0}d_{j}), \\
  \tr_{q^m/q}(\alpha^{1}d_{i})=\tr_{q^m/q}(\alpha^{1}d_{j}), \\
  \vdots \\
  \tr_{q^m/q}(\alpha^{m-1}d_{i})=\tr_{q^m/q}(\alpha^{m-1}d_{j}),
\end{array}\right.
\end{eqnarray}
where $d_i, d_j \in D$ and $d_i \neq d_j$.
For any $x=\sum_{l=0}^{m-1}k_l\alpha^l \in \gf_{q^m}, k_l \in \gf_q$, the System (\ref{system2}) implies
$$
\tr_{q^m/q}(x(d_i-d_j))=\sum_{l=0}^{m-1}k_l \tr_{q^m/q}(\alpha^l(d_i-d_j))=0.
$$
This contradicts with the fact that $|\ker(\tr_{q^m/q})|=q^{m-1}$.
Hence, $\bg_i$ and $\bg_j$ are $\gf_q$-linearly independent and $d^{\perp} \geq 3$.
Besides, there exist three columns which are $\gf_q$-linearly dependent by the discussions above.
Then $d^{\perp}=3$.
\end{proof}

\begin{remark}
In Theorem \ref{loc}, if $q=2$, it is easy to verify that the locality of $\cC_D$ is at least $3$.
However, it is difficult to explicitly determine the locality of $\cC_D$ and the minimal distance of $\cC_D^\perp$ for $q=2$.
\end{remark}

\subsection{The case that $N=q^r+1$}
Let $q$ be a power of a prime $p$. Let $s, m$ be two positive integers with $s \mid m$ and $\Norm_{q^m/q^s}$ the norm function from $\gf_{q^m}$ to $\gf_{q^s}$ defined by
$\Norm_{q^m/q^s}(x)=x\cdot x^{q^s}\cdot x^{q^{2s}} \cdots x^{q^{(m/s-1)s}}=x^{\frac{q^m-1}{q^s-1}}, x \in \gf_{q^m}.$
In this subsection, let $q$ be a prime power, $m, r$ be two positive integers such that $m=2r$ and $r \geq 2$. Then $\Norm_{q^m/q^r}(x)=x^{q^r+1}$.
Let $N=q^r+1$. Then the defining set is given by
\begin{eqnarray}\label{eq-Dr}
  D=\{x \in \gf_{q^m}: \tr_{q^r/q}(x^{q^r+1})=0\}.
\end{eqnarray}
Firstly, we will determine the parameters and weight distribution of $\overline{\cC_D}$ and prove that $\overline{\cC_D}$ is a family of self-orthogonal codes. Then we will prove that  $\overline{\cC_D}$ is optimally extendable or almost optimally extendable.

\begin{lemma}\label{lem-length}
The length $n$ of $\overline{\cC_{D}}$ is given by $n=q^{r-1}(q^r-q+1)$.
\end{lemma}
\begin{proof}
Since the trace function $\tr_{q^r/q}$  is  a homomorphism from $\gf_{q^r}$ to $\gf_q$ and the norm function $\Norm_{q^m/q^r}$ is also a homomorphism from $\gf_{q^m}^*$ to $\gf_{q^r}^*$, the conclusion follows.
\end{proof}

\begin{lemma}\label{lem-Ns}
Let $s \in \gf_q^*$ and $b \in \gf_{q^m}^*$, where  $m=2r$ and $q$ is a prime power. Let $N_s$ denote the number of solutions in $\gf_{q^m}$ of the following system of equations as
\begin{eqnarray*}\label{eq-01}
\left\{
\begin{array}{l}
\tr_{q^m/q}\left(bx\right)=s, \\
\tr_{q^r/q}(x^{q^r+1})=0.
\end{array}
\right.
\end{eqnarray*}
Then
\begin{eqnarray*}
N_s=
\begin{cases}
q^{m-2}     &\text{if $\tr_{q^r/q}\left(b^{q^r+1}\right)=0$,}\\
q^{m-2}-q^{r-1}    &\text{if $\tr_{q^r/q}\left(b^{q^r+1}\right)\neq0$.}
\end{cases}
\end{eqnarray*}

\end{lemma}

\begin{proof}
 By the orthogonal relation of additive characters, we have
\begin{eqnarray}\label{Ns}
N_s &=&| \{x\in\gf_{q^m}: \tr_{q^m/q}\left(bx\right)=s, \tr_{q^r/q}(x^{q^r+1})=0\}|\nonumber\\
&=& \frac{1}{q^2}\sum_{u\in\gf_q}\sum_{v\in\gf_q}\sum_{x\in\gf_{q^m}}\zeta_p^{\tr_{q/p}\left(\tr_{q^m/q}\left(ubx\right)-us\right)}
\zeta_p^{\tr_{q/p}\left(\tr_{q^r/q}\left(vx^{q^r+1}\right)\right)}\nonumber\\
&=&\frac{1}{q^2}\left(\sum_{u\in\gf_q^*}\sum_{x\in\gf_{q^m}}\chi_1\left(ubx\right)\phi_1\left(-us\right) +
\sum_{v\in\gf_q^*}\sum_{x\in\gf_{q^m}}\zeta_p^{\tr_{q^r/p}\left(vx^{q^r+1}\right)} + \right.\nonumber\\
&&\left.\sum_{v\in\gf_q^*}\sum_{u\in\gf_q^*}\phi_1(-us)\sum_{x\in\gf_{q^m}}\zeta_p^{\tr_{q^m/p}(ubx)+\tr_{q^r/p}\left(vx^{q^r+1}\right)}+ q^m\right).
\end{eqnarray}
By the orthogonal relation of additive characters, we have
\begin{eqnarray}\label{Ns1}
S_1:=\sum_{u\in\gf_q^*}\sum_{x\in\gf_{q^m}}\chi_1\left(ubx\right)\phi_1\left(-us\right)
=\sum_{u\in\gf_q^*}\phi_1\left(-us\right)\sum_{x\in\gf_{q^m}}\chi_1\left(ubx\right)=0.
\end{eqnarray}
By Lemma \ref{lem-2k},
\begin{eqnarray}\label{Ns2}
S_2:=\sum_{v\in\gf_q^*}\sum_{x\in\gf_{q^m}}\zeta_p^{\tr_{q^r/p}\left(vx^{q^r+1}\right)}
=\sum_{v\in\gf_q^*}\left(-q^r\right)
=-q^r\left(q-1\right),
\end{eqnarray}
and
\begin{eqnarray}\label{Ns3}
S_3&:=&\sum_{v\in\gf_q^*}\sum_{u\in\gf_q^*}\phi_1(-us)\sum_{x\in\gf_{q^m}}\zeta_p^{\tr_{q^m/p}(ubx)+\tr_{q^r/p}\left(vx^{q^r+1}\right)}\nonumber\\
&=&\sum_{v\in\gf_q^*}\sum_{u\in\gf_q^*}\phi_1\left(-us\right)(-q^r)\zeta_p^{-\tr_{q^r/p}\left(\frac{u^2b^{q^r+1}}{v}\right)}\nonumber\\
&=&-q^r\sum_{v\in\gf_q^*}\sum_{u\in\gf_q^*}\phi_1\left(-us\right)\phi_1\left(\frac{-\tr_{q^r/q}\left(b^{q^r+1}\right)}{v}u^2\right)\nonumber\\
&=&-q^r\sum_{v\in\gf_q^*}\sum_{u\in\gf_q^*}\phi_1\left(\frac{-\tr_{q^r/q}\left(b^{q^r+1}\right)}{v}u^2-su\right)\nonumber\\
&=&\begin{cases}
q^r(q-1) &\text{if $\tr_{q^r/q}\left(b^{q^r+1}\right)=0$,}\\
-q^r\sum\limits_{u\in\gf_q^*}\phi_1(-su)\sum\limits_{v\in\gf_q^*}\phi_1\left(\frac{-\tr_{q^k/q}\left(b^{q^k+1}\right)u^2}{v}\right)     &\text{if $\tr_{q^r/q}\left(b^{q^r+1}\right)\neq0$,}
\end{cases}\nonumber\\
&=&\begin{cases}
q^r(q-1) &\text{if $\tr_{q^r/q}\left(b^{q^r+1}\right)=0$,}\\
-q^r     &\text{if $\tr_{q^r/q}\left(b^{q^r+1}\right)\neq0$.}
\end{cases}
\end{eqnarray}
Substituting Equations (\ref{Ns1}), (\ref{Ns2}) and (\ref{Ns3}) into Equation (\ref{Ns}) yields the desired conclusion.
\end{proof}

\begin{theorem}\label{tem-wtCD1}
Let $q = p^e$, where $p$ is a prime and $e$ is a positive integer. Let $m, r$ be positive integers such that $m=2r$ with $r \geq 2$. Then the augmented code $\overline{\cC_{D}}$ with the defining set $D=\{x \in \gf_{q^m} : \tr_{q^r/q}(\Norm_{q^m/q^r}(x))=0\}$ has parameters $[q^{r-1}(q^r-q+1), m+1, q^{r-1}(q^{r}-q^{r-1}-q+1)]$ and its weight distribution is listed in Table \ref{tab2}. Besides, if $q > 2$, $\overline{\cC_{D}}^{\perp}$ has parameters $[q^{r-1}(q^r-q+1), q^{r-1}(q^r-q+1)-m-1, 3]$ and is at least almost optimal according to the sphere-packing bound. If $q=2$ and $r > 2$, $\overline{\cC_{D}}^{\perp}$ has parameters $[q^{r-1}(q^r-q+1), q^{r-1}(q^r-q+1)-m-1, 4]$ and is optimal according to the sphere-packing bound. 
\begin{table}[h!]
\begin{center}
\caption{The weight distribution of $\overline{\cC_{D}}$ in Theorem \ref{tem-wtCD1}.}\label{tab2}
\begin{tabular}{@{}ll@{}}
\toprule
Weight & Frequency  \\
\midrule
$0$ & $1$ \\
$q^{r-1}(q^{r}-q^{r-1}-q+1)$ &  $q^{r-1}(2q^{r+1}-2q^r-q^2+3q-2)-q+1$ \\
$q^{r-1}(q^r-q^{r-1}-q+2)$ &  $q^{r-1}(q-1)(q^{r+1}-q^r+q-1)$ \\
$q^{2r-2}(q-1)$ &  $q^{r-1}(q^{r}-q+1)-1$ \\
$q^{r-1}(q^r-q+1)$ &  $q-1$ \\
\bottomrule
\end{tabular}
\end{center}
\end{table}
\end{theorem}

\begin{proof}
The length of $\overline{\cC_{D}}=\{\bc_{(b,c)}=(\tr_{q^m/q}(bx))_{x \in D} + c\mathbf{1}:b \in \gf_{q^m}, c \in \gf_q\}$ is $n=q^{r-1}(q^r-q+1)$ by Lemma \ref{lem-length}.

For $b=0$, it is obvious that the codeword $\bc_{(b,c)}$ has Hamming weight
\begin{eqnarray*}
\text{wt}(\bc_{(b,c)})=
\begin{cases}
0,&\text{if $c=0$,}\\
q^{m-1}-q^{r-1}(q-1),     &\text{if $c \neq 0$.}
\end{cases}
\end{eqnarray*}

For $b \in \gf_q^*$, we denote by $N_{(b,c)}=|\{x \in D:\tr_{q^m/q}(bx)+c=0\}|=|\{x \in \gf_{q^m}:\tr_{q^r/q}\left(x^{q^r+1}\right)=0 \mbox{ and } \tr_{q^m/q}(bx)+c=0\}|.$ Then the Hamming weight of a codeword $\bc_{(b,c)}$ is $\text{wt}(\bc_{(b,c)})=n-N_{(b,c)}$ for $b \in \gf_q^*, c \in \gf_q$. For $c \in \gf_q^*$, we can find that the value of $N_{(b,c)}$ is equal to the value of $N_s$ in Lemma \ref{lem-Ns} for $s=-c$. Then by Lemma \ref{lem-Ns}, we have
\begin{eqnarray*}
\text{wt}(\bc_{(b,c)})=n-N_{(b,c)}
=\begin{cases}
q^{m-1}-q^{m-2}-q^{r-1}(q-1)     &\text{if $\tr_{q^r/q}\left(b^{q^r+1}\right)=0$,}\\
q^{m-1}-q^{m-2}-q^r+2q^{r-1}     &\text{if $\tr_{q^r/q}\left(b^{q^r+1}\right)\neq0$.}
\end{cases}
\end{eqnarray*}
For $c=0$, by \cite[Theorem 1]{HZ4}, we have
\begin{eqnarray*}
\text{wt}(\bc_{(b,c)})&=&\begin{cases}
q^{m-1}-q^{m-2} &\text{if $\tr_{q^r/q}\left(b^{q^r+1}\right)=0$,}\\
q^{m-1}-q^{m-2}-q^{r-1}(q-1)     &\text{if $\tr_{q^r/q}\left(b^{q^r+1}\right)\neq0$.}
\end{cases}
\end{eqnarray*}

Summarizing the above discussions, we derive that
\begin{eqnarray*}
\text{wt}(\bc_{(b,c)})&=&\begin{cases}
0 & \mbox{if $b=0$ and $c=0$,}\\
q^{m-1}-q^{r-1}(q-1) & \mbox{if $b=0$ and $c \neq 0$,}\\
q^{m-1}-q^{m-2} & \mbox{if $b \in \gf_{q^m}^*, \tr_{q^r/q}\left(b^{q^r+1}\right)=0$ and $c=0$,}\\
q^{m-1}-q^{m-2}-q^{r-1}(q-1)   & \mbox{if $b \in \gf_{q^m}^*, \tr_{q^r/q}\left(b^{q^r+1}\right)=0$ and $c\neq0$,}\\
 & \mbox{or $\tr_{q^r/q}\left(b^{q^r+1}\right)\neq0$ and $c=0$,}\\
q^{m-1}-q^{m-2}-q^r+2q^{r-1} & \mbox{if $\tr_{q^r/q}\left(b^{q^r+1}\right)\neq0$ and $c\neq0$.}
\end{cases}
\end{eqnarray*}

Denote by $N_b=|\{b \in \gf_{q^m}^*:\tr_{q^r/q}\left(b^{q^r+1}\right)=0\}|=|\{b \in \gf_{q^m}^*:\tr_{q^r/q}\left(\Norm_{q^m/q^r}(b)\right)=0\}|$. Obviously, $N_b=(q^{r-1}-1)(q^r+1)=q^{r-1}(q^{r}-q+1)-1$ as the trace and norm functions are both surjective homomorphisms.
 Then the weight distribution can be given in Table \ref{tab2}.

Finally, we will determine the minimum distance of $\overline{\cC_{D}}^{\perp}$. Note that the weight enumerator of $\overline{\cC_{D}}$ is
\begin{eqnarray}
A(z)=1+A_{w_1}z^{q^{r-1}(q^{r}-q^{r-1}-q+1)}+A_{w_2}z^{q^{r-1}(q^r-q^{r-1}-q+2)}+
A_{w_3}z^{q^{2r-2}(q-1)}+A_{w_4}z^{q^{r-1}(q^r-q+1)},
\end{eqnarray}
where $A_{w_1}$, $A_{w_2}$, $A_{w_3}$ and $A_{w_4}$ are given in Table \ref{tab2}. Besides, by the second, third and fourth Pless power moments in \cite[Page 260]{H}, we have
\begin{eqnarray*}
A_1^{\perp}=A_2^{\perp}=0,\ A_3^{\perp}=\frac{q^r(q^r+1)(q-1)(q-2)(q^r-q^2+q)(q^r-q)}{6q^3}.
\end{eqnarray*}
When $q = 2$, we have $A_3^{\perp} = 0$. By the fifth Pless power moment in \cite[Page 260]{H}, we have
$$A_4^{\perp}=\frac{2^r(2^r-1)(2^r-2)(2^r-4)(2^r+2)(2^r+1)}{384}.$$
Then $A_4^{\perp} > 0$ for $r >2$ and  $A_4^{\perp} = 0$ for $r = 2$. Hence, $\overline{\cC_{D}}^{\perp}$ has parameters $[q^{r-1}(q^r-q+1), q^{r-1}(q^r-q+1)-m-1, 4]$ when $q=2,r>2$. 
When $q > 2$, it is easy to deduce that $A_3^{\perp} > 0$ as $r \geq 2$. In this case, $\overline{\cC_{D}}^{\perp}$ has parameters $[q^{r-1}(q^r-q+1), q^{r-1}(q^r-q+1)-m-1, 3]$.
The proof is completed by the sphere packing bound in Lemma \ref{sphere}.
\end{proof}

To prove that $\overline{\cC_{D}}$ is self-orthogonal, we determine the following system of equations at first.

\begin{lemma}\label{lem-Nst}
Let $m=2r$ and $(s,t) \in \gf_q^* \times \gf_q^*$ with $q$ a power of $2$. Let $N_{s,t}$ denote the number of solutions in $\gf_{q^m}$ of the following system of equations:
\begin{eqnarray*}\label{eq-03}
\left\{
\begin{array}{l}
\tr_{q^m/q}\left(bx\right)=s, \\
\tr_{q^m/q}\left(b'x\right)=t, \\
\tr_{q^r/q}(x^{q^r+1})=0,
\end{array}
\right.
\end{eqnarray*}
where $b,b' \in \gf_{q^m}^*$. If $b$ and $b'$ are $\gf_q$-linearly independent, then the value of $N_{s,t}$ is listed in Table \ref{tab1}, where $$a_1=\tr_{q^r/q}\left(b'^{q^r+1}\right), a_2=\tr_{q^r/q}\left(b^{q^r+1}\right), a_3=\tr_{q^r/q}\left(bb'^{q^r}+b'b^{q^r}\right) \mbox{ and } \Delta=\tr_{q/2}\left(\frac{a_1a_2}{a_3^2}\right).$$
\begin{table}[!h]
\begin{center}
\caption{The value of $N_{s,t}$ in Lemma \ref{lem-Nst}. \label{tab1}}
\begin{tabular}{ll}
\toprule
 $N_{s,t}$ &   Condition  \\
\midrule
$q^{2r-3}$& for $a_1=0$ and $a_3=0$ \\
$q^{r-2}\left(q^{r-1}+1\right)$ & for $a_1=0$, $a_2=0$ and $a_3\neq0$ \\
$q^{r-1}\left(q^{r-2}-1\right)$     &for $a_1=0$, $a_2\neq0$, $a_3\neq0$ and $ta_2+sa_3=0$\\
$q^{2r-3}$     &for $a_1=0$, $a_2\neq0$, $a_3\neq0$ and $ta_2+sa_3 \neq 0$\\
$q^{2r-3}$     &for $a_1\neq0$, $a_2=0$ and $a_3=0$\\
$q^{r-1}\left(q^{r-2}-1\right)$     &for $a_1\neq0$, $a_2=0$, $a_3\neq0$ and $sa_1+ta_3=0$\\
$q^{2r-3}$     &for $a_1\neq0$, $a_2=0$, $a_3\neq0$ and $sa_1+ta_3\neq 0$\\
$q^{r-2}\left(q^{r-1}-1\right)$    &for $a_1\neq0$, $a_2\neq0$, $a_3=0$ and $s^2a_1=t^2a_2$\\
$q^{2r-3}$    &for $a_1\neq0$, $a_2\neq0$, $a_3=0$ and $s^2a_1 \neq t^2a_2$\\
$q^{r-2}(q^{r-1}-(q-1)(-1)^{\Delta}-1)$   &for $a_1\neq0$, $a_2\neq0$, $a_3\neq0$ and $a_2t^2+a_3st+a_1s^2=0$ \\
$q^{r-2}(q^{r-1}+(-1)^{\Delta}-1)$   &for $a_1\neq0$, $a_2\neq0$, $a_3\neq0$ and $a_2t^2+a_3st+a_1s^2\neq0$ \\
\bottomrule
\end{tabular}
\end{center}
\end{table}
\end{lemma}

\begin{proof}
By the orthogonal relation of additive characters, we have
\begin{eqnarray}\label{Nst}
N_{s,t} &=& |\{x\in\gf_{q^m}: \tr_{q^m/q}\left(bx\right)=s, \tr_{q^m/q}\left(b'x\right)=t, \tr_{q^r/q}(x^{q^r+1})=0\}|\nonumber\\
&=&\frac{1}{q^3}\sum_{u\in\gf_q}\sum_{v\in\gf_q}\sum_{w\in\gf_q}\sum_{x\in\gf_{q^m}}\zeta_p^{\tr_{q/p}\left(\tr_{q^m/q}\left(ubx\right)-us\right)}
\zeta_p^{\tr_{q/p}\left(\tr_{q^m/q}\left(vb'x\right)-vt\right)}\zeta_p^{\tr_{q/p}\left(\tr_{q^r/q}\left(wx^{q^r+1}\right)\right)}\nonumber\\
&=&\frac{1}{q^3}\left(q^m+\sum_{u\in\gf_q^*}\sum_{x\in\gf_{q^m}}\chi_1\left(ubx\right)\phi_1\left(-us\right) +
\sum_{v\in\gf_q^*}\sum_{x\in\gf_{q^m}}\phi_1\left(-vt\right)\chi_1\left(vb'x\right) + \right.\nonumber\\
&&\left.\sum_{w\in\gf_q^*}\sum_{x\in\gf_{q^m}}\zeta_p^{\tr_{q^r/p}\left(wx^{q^r+1}\right)} +\sum_{v\in\gf_q^*}\sum_{u\in\gf_q^*}\sum_{x\in\gf_{q^m}}\phi_1\left(-us-vt\right)\chi_1\left(ubx+vb'x\right)+\right.\nonumber\\
&&\left.\sum_{u\in\gf_q^*}\sum_{w\in\gf_q^*}\sum_{x\in\gf_{q^m}}\phi_1\left(-us\right)\chi_1\left(ubx\right)\zeta_p^{\tr_{q^r/p}\left(wx^{q^r+1}\right)}+\right.\nonumber\\
&&\left.\sum_{v\in\gf_q^*}\sum_{w\in\gf_q^*}\sum_{x\in\gf_{q^m}}\phi_1\left(-vt\right)\chi_1\left(vb'x\right)\zeta_p^{\tr_{q^r/p}\left(wx^{q^r+1}\right)}+\right.\nonumber\\
&&\left.\sum_{u\in\gf_q^*}\sum_{v\in\gf_q^*}\sum_{w\in\gf_q^*}\sum_{x\in\gf_{q^m}}\phi_1\left(-us-vt\right)\chi_1\left(ubx+vex\right)\zeta_p^{\tr_{q^r/p}\left(wx^{q^r+1}\right)}\right).
\end{eqnarray}
By the orthogonal relation  of the additive characters, we have
\begin{eqnarray}\label{N1}
N_1:=\sum_{u\in\gf_q^*}\sum_{x\in\gf_{q^m}}\chi_1\left(ubx\right)\phi_1\left(-us\right)=0
\end{eqnarray}
and
\begin{eqnarray}\label{N2}
N_2:=\sum_{v\in\gf_q^*}\sum_{x\in\gf_{q^m}}\chi_1\left(vb'x\right)\phi_1\left(-vt\right)=0.
\end{eqnarray}
By Lemma \ref{lem-2k},
\begin{eqnarray}\label{N3}
N_3:=\sum_{w\in\gf_q^*}\sum_{x\in\gf_{q^m}}\zeta_p^{\tr_{q^r/p}\left(wx^{q^r+1}\right)}=-q^r(q-1).
\end{eqnarray}
By the orthogonal relation  of the additive characters,
\begin{eqnarray}\label{N4}
\nonumber N_4&:=&\sum_{u\in\gf_q^*}\sum_{v\in\gf_q^*}\sum_{x\in\gf_{q^m}}\phi_1\left(-us-vt\right)\chi_1\left(ubx+vb'x\right)\\
\nonumber &=&\sum_{u\in\gf_q^*}\sum_{v\in\gf_q^*}\phi_1\left(-us-vt\right)\sum_{x\in\gf_{q^m}}\chi_1\left((ub+vb')x\right)\\
 &=&0,
\end{eqnarray}
where $ub+vb'\neq 0$ as $b$ and $b'$ are $\gf_q$-linearly independent. Similarly to the calculation of Equation $(\ref{Ns3})$, we can obtain
\begin{eqnarray}\label{N5}
\nonumber N_5&:=&\sum_{u\in\gf_q^*}\sum_{w\in\gf_q^*}\sum_{x\in\gf_{q^m}}\phi_1\left(-us\right)\chi_1\left(ubx\right)\zeta_p^{\tr_{q^r/p}\left(wx^{q^r+1}\right)}\\
&=&\begin{cases}
q^r(q-1) &\text{if $\tr_{q^r/q}\left(b^{q^r+1}\right)=0$,}\\
-q^r     &\text{if $\tr_{q^r/q}\left(b^{q^r+1}\right)\neq0$,}
\end{cases}
\end{eqnarray}
and
\begin{eqnarray}\label{N6}
\nonumber N_6&:=&\sum_{v\in\gf_q^*}\sum_{w\in\gf_q^*}\sum_{x\in\gf_{q^m}}\phi_1\left(-vt\right)\chi_1\left(vb'x\right)\zeta_p^{\tr_{q^r/p}\left(wx^{q^r+1}\right)}\\
&=&\begin{cases}
q^r(q-1) &\text{if $\tr_{q^r/q}\left(b'^{q^r+1}\right)=0$,}\\
-q^r     &\text{if $\tr_{q^r/q}\left(b'^{q^r+1}\right)\neq0$.}
\end{cases}
\end{eqnarray}
By Lemma \ref{lem-2k},
\begin{eqnarray}\label{N7}
N_7&:=&\sum_{u\in\gf_q^*}\sum_{v\in\gf_q^*}\sum_{w\in\gf_q^*}\sum_{x\in\gf_{q^m}}\phi_1\left(-us-vt\right)\chi_1\left(ubx+vb'x\right)\zeta_p^{\tr_{q^r/p}\left(wx^{q^r+1}\right)}\nonumber\\
&=&\sum_{u\in\gf_q^*}\sum_{v\in\gf_q^*}\sum_{w\in\gf_q^*}\phi_1\left(-us-vt\right)\sum_{x\in\gf_{q^m}}\zeta_p^{\tr_{q^r/p}\left(wx^{q^r+1}\right)+\tr_{q^m/p}\left((ub+vb')x\right)}\nonumber\\
&=&-q^r\sum_{u\in\gf_q^*}\sum_{v\in\gf_q^*}\sum_{w\in\gf_q^*}\phi_1\left(-us-vt\right)\zeta_p^{-\tr_{q^r/p}\left(\frac{(ub+vb')^{q^r+1}}{w}\right)}\nonumber\\
&=&-q^r\sum_{u\in\gf_q^*}\sum_{v\in\gf_q^*}\sum_{w\in\gf_q^*}\phi_1\left(-us-vt\right)\phi_1\left(-\frac{1}{w}\left(\tr_{q^r/q}\left(b^{q^r+1}\right)u^2+\right.\right.\nonumber\\
&&\left.\left.\tr_{q^r/q}\left(bb'^{q^r}+b'b^{q^r}\right)uv+\tr_{q^r/q}\left(b'^{q^r+1}\right)v^2\right)\right).
\end{eqnarray}
From now on, let $a_1=\tr_{q^r/q}\left(b'^{q^r+1}\right)$, $a_2=\tr_{q^r/q}\left(b^{q^r+1}\right)$ and $a_3=\tr_{q^r/q}\left(bb'^{q^r}+b'b^{q^r}\right)$. In order to calculate $N_7$, we consider the following cases.
\begin{enumerate}[1.]
\item If $a_1=0, a_2=0$ and $a_3=0,$ then
\begin{eqnarray*}
N_7=-q^r\sum_{u\in\gf_q^*}\sum_{v\in\gf_q^*}\sum_{w\in\gf_q^*}\phi_1\left(-us-vt\right)=-q^r(q-1).
\end{eqnarray*}

\item If $a_1=0, a_2=0$ and $a_3\neq 0,$ then
\begin{eqnarray*}
N_7=-q^r\sum_{u\in\gf_q^*}\phi_1\left(-us\right)\sum_{v\in\gf_q^*}\phi_1\left(-vt\right)\sum_{w\in\gf_q^*}\phi_1\left(-\frac{a_3uv}{w}\right)=q^r.
\end{eqnarray*}

\item If $a_1=0, a_2\neq0$ and $a_3=0,$ then
\begin{eqnarray*}
N_7&=&-q^r\sum_{u\in\gf_q^*}\left(-us\right)\sum_{v\in\gf_q^*}\phi_1\left(-vt\right)\sum_{w\in\gf_q^*}\phi_1\left(-\frac{a_2u^2}{w}\right)
=q^r.
\end{eqnarray*}

\item If $a_1=0, a_2\neq0$ and $a_3\neq0,$ then
\begin{eqnarray*}
N_7&=&-q^r\sum_{u\in\gf_q^*}\sum_{v\in\gf_q^*}\sum_{w\in\gf_q^*}\phi_1\left(-us-vt\right)\phi_1\left(-\frac{a_2u^2}{w}-
\frac{a_3uv}{w}\right)\\
&=&-q^r\sum_{v\in\gf_q^*}\sum_{w\in\gf_q^*}\phi_1\left(-vt\right)\sum_{u\in\gf_q}\phi_1\left(-\frac{a_2}{w}u^2-\left(\frac{a_3v}{w}+s\right)u\right)
+q^r\sum_{v\in\gf_q^*}\sum_{w\in\gf_q^*}\phi_1\left(-vt\right)\\
&=&-q^{r+1}\sum_{\substack{w, v\in\gf_q^*\\-\frac{a_2}{w}=\left(\frac{a_3v}{w}+s\right)^2}}\phi_1(-vt)-q^r(q-1)\\
&=& -q^{r+1}\sum_{w \in \gf_q^*}\phi_1\left(\frac{sw+a_2^{\frac{q}{2}}w^{\frac{q}{2}}}{a_3}t\right)-q^r(q-1)\\
&=& -q^{r+1}\sum_{w \in \gf_q^*}\phi_1\left(\frac{a_2^{\frac{q}{2}}t}{a_3}w^{\frac{q}{2}}+\frac{st}{a_3}w\right)-q^r(q-1)\\
&=& \begin{cases}
      -q^{r+1}(q-1)-q^r(q-1) & \mbox{if } ta_2+sa_3=0 \\
      q^{r+1}-q^r(q-1) & \mbox{otherwise}
    \end{cases}\\
&=& \begin{cases}
      -q^r(q^2-1) & \mbox{if } ta_2+sa_3=0, \\
      q^r & \mbox{otherwise},
    \end{cases}
\end{eqnarray*}
where the third and sixth equalities hold due to Lemma \ref{lem-sums} and the fourth equality holds due to the fact that
the solution $v$ of the equation $-\frac{a_2}{w}=\left(\frac{a_3v}{w}+s\right)^2$ is unique by Lemma \ref{lem-root} and $v=\frac{sw+a_2^{\frac{q}{2}}w^{\frac{q}{2}}}{a_3}$.

\item If $a_1\neq0, a_2=0$ and $a_3=0,$ then
\begin{eqnarray*}
N_7&=&-q^r\sum_{u\in\gf_q^*}\left(-us\right)\sum_{v\in\gf_q^*}\phi_1\left(-vt\right)\sum_{w\in\gf_q^*}\phi_1\left(-\frac{a_1v^2}{w}\right)=q^r.
\end{eqnarray*}

\item If $a_1 \neq0, a_2 =0$ and $a_3 \neq0,$ then
\begin{eqnarray*}
N_7&=& \begin{cases}
      -q^r(q^2-1) & \mbox{if } sa_1+ta_3=0, \\
      q^r & \mbox{otherwise},
    \end{cases}
\end{eqnarray*}
by the similar calculation method in Case 4.

\item If $a_1 \neq0, a_2 \neq0$ and $a_3=0,$ then
\begin{eqnarray*}
N_7&=&-q^r\sum_{u\in\gf_q^*}\sum_{v\in\gf_q^*}\sum_{w\in\gf_q^*}\phi_1\left(-us-vt\right)\phi_1\left(-\frac{a_2u^2}{w}-
\frac{a_1v^2}{w}\right)\\
&=&-q^r\sum_{u\in\gf_q^*}\sum_{v\in\gf_q^*}\sum_{w\in\gf_q^*}\phi_1\left(-\frac{a_2}{w}u^2-su\right)
\phi_1\left(-\frac{a_1}{w}v^2-tv\right)\\
&=&-q^r\sum_{w\in\gf_q^*}\sum_{u\in\gf_q}\phi_1\left(-\frac{a_2}{w}u^2-su\right)
\sum_{v\in\gf_q}\phi_1\left(-\frac{a_1}{w}v^2-tv\right)+
q^r\sum_{w\in\gf_q^*}\sum_{v\in\gf_q^*}\phi_1\left(-\frac{a_1}{w}v^2
-tv\right)+\\
&&q^r\sum_{w\in\gf_q^*}\sum_{u\in\gf_q^*}\phi_1\left(-\frac{a_2}{w}u^2-su\right)+q^r(q-1)\\
&=&\begin{cases}
     -q^{r+1}+q^r+q^r+q^r(q-1) & \mbox{if } s^2a_1=t^2a_2 \\
     0+q^r+q^r+q^r(q-1) & \mbox{otherwise}
   \end{cases}\\
&=&\begin{cases}
     q^r & \mbox{if } s^2a_1=t^2a_2, \\
     q^r(q+1) & \mbox{otherwise},
   \end{cases}
\end{eqnarray*}
where the fourth equality holds due to Lemma \ref{lem-sums}.

\item If $a_1 \neq 0, a_2 \neq 0$ and $a_3 \neq 0,$ then
\begin{eqnarray*}
N_7&=&-q^r\sum_{u\in\gf_q^*}\sum_{v\in\gf_q^*}\sum_{w\in\gf_q^*}\phi_1\left(-us-vt\right)\phi_1\left(-\frac{1}{w}(a_2u^2+
a_3uv+a_1v^2)\right)\\
&=&-q^r\sum_{w\in\gf_q^*}\sum_{v\in\gf_q^*}\phi_1\left(-vt-\frac{a_1}{w}v^2\right)
\sum_{u\in\gf_q}\phi_1\left(-\frac{a_2}{w}u^2-\left(\frac{va_3}{w}+s\right)u\right)+\\
&&q^r\sum_{w\in\gf_q^*}\sum_{v\in\gf_q^*}\phi_1\left(-vt-\frac{a_1}{w}v^2\right)\\
&=&q^r-q^r\sum_{\substack{w,v\in\gf_q^*\\-\frac{a_2}{w}=\left(\frac{a_3v}{w}+s\right)^2}}\phi_1\left(-vt-\frac{a_1}{w}v^2\right)\\
&=&q^r-q^{r+1}\sum_{w \in \gf_q^*}\phi_1\left(\frac{a_2^{\frac{q}{2}}t}{a_3}w^{\frac{q}{2}}+\left(\frac{st}{a_3}+\frac{a_1s^2}{a_3^2}\right)w+\frac{a_1a_2}{a_3^2}\right)\\
&=&\begin{cases}
     q^r-q^{r+1}(q-1)\phi_1\left(\frac{a_1a_2}{a_3^2}\right) & \mbox{if } a_2t^2+a_3st+a_1s^2=0 \\
     q^r-q^{r+1}\left(0-\phi_1\left(\frac{a_1a_2}{a_3^2}\right)\right) & \mbox{if } a_2t^2+a_3st+a_1s^2 \neq 0
   \end{cases}\\
&=&\begin{cases}
     q^r-q^{r+1}(q-1)(-1)^{\tr_{q/2}\left(\frac{a_1a_2}{a_3^2}\right)} & \mbox{if } a_2t^2+a_3st+a_1s^2=0, \\
     q^r+q^{r+1}(-1)^{\tr_{q/2}\left(\frac{a_1a_2}{a_3^2}\right)} & \mbox{if } a_2t^2+a_3st+a_1s^2 \neq 0,
   \end{cases}
\end{eqnarray*}
where the third and fifth equalities hold due to Lemma \ref{lem-sums} and the fourth equality holds due to the fact that
the solution $v$ of the equation $-\frac{a_2}{w}=\left(\frac{a_3v}{w}+s\right)^2$ is unique by Lemma \ref{lem-root} and $v=\frac{sw+a_2^{\frac{q}{2}}w^{\frac{q}{2}}}{a_3}$. Note that $\phi_1\left(\frac{a_1a_2}{a_3^2}\right)=\zeta_2^{\tr_{q/2}\left(\frac{a_1a_2}{a_3^2}\right)}=(-1)^{\tr_{q/2}\left(\frac{a_1a_2}{a_3^2}\right)}=\pm1$.
\end{enumerate}

Substituting Equations (\ref{N1}), (\ref{N2}), (\ref{N3}), (\ref{N4}), (\ref{N5}), (\ref{N6}) and (\ref{N7}) into Equation (\ref{Nst}) yields the values of $N_{s,t}$  listed in Table \ref{tab1}.
\end{proof}

\begin{theorem}\label{tem-so1}
Let $q = p^e$, where $p$ is a prime and $e$ is a positive integer. Let $m, r$ be positive integers such that $m=2r$ with $r \geq 2$ and $(q,r) \neq (2,2)$. Then the augmented code $\overline{\cC_{D}}$ with the defining set $D=\{x \in \gf_{q^m} : \tr_{q^r/q}(x^{q^r+1})=0\}$ is self-orthogonal.
\end{theorem}

\begin{proof}
  By definition, $\overline{\cC_{D}}=\{\mathbf{c}_{(b,c)}=(\tr_{q^m/q}(bx) + c)_{x \in {D}}:b \in \gf_{q^m}, c \in \gf_q\}$ and $\mathbf{1} \in \overline{\cC_{D}}$. Furthermore, we derive that $\overline{\cC_{D}}$ is $p$-divisible by Theorem \ref{tem-wtCD1}. Then by Lemma \ref{lem-self-orthogonal}, we have $\overline{\cC_{D}}$ is self-orthogonal for odd $q$.

If $q$ is even, we need to prove that $\bc_{(b,c)}\cdot \bc_{(b',c')} = 0$ for any $(b,c) \in \gf_{q^m} \times \gf_q$, $(b',c') \in \gf_{q^m} \times \gf_q$.
In the following, we consider some cases.

{Case 1:} Let $b=0,c \in \gf_q$, $b'=0, c' \in \gf_q$. Then
\begin{eqnarray*}
\bc_{(0,c)}\cdot \bc_{(0,c')} &=& \sum_{x\in D}cc' = \left(q^{m-1}-q^{r-1}(q-1)\right)cc'
= 0
\end{eqnarray*}
as $r\geq 2$ and $|D|=q^{r-1}(q^r-q+1)$ by Lemma \ref{lem-length}.

{Case 2:} Let $b=0,c \in \gf_q$, $b'\neq 0, c' \in \gf_q$. Then
\begin{eqnarray*}
\bc_{(0,c)}\cdot \bc_{(b',c')} &=& c\sum_{x\in D}\left(\tr_{q^m/q}\left(b'x\right)+c'\right) \\
&=&  c\sum_{x\in D}\tr_{q^m/q}\left(b'x\right)+\sum_{x\in D}cc'\\
&=&  c\sum_{s\in \gf_q}N_s s + \left(q^{m-1}-q^{r-1}(q-1)\right)cc'\\
&=&  c\sum_{s\in \gf_q^*}N_s s + \left(q^{m-1}-q^{r-1}(q-1)\right)cc'\\
&=&cN_s\sum_{s\in \gf_q^*} s + \left(q^{m-1}-q^{r-1}(q-1)\right)cc'=0,
\end{eqnarray*}
where $N_s$ was defined as in Lemma \ref{lem-Ns} and is independent on $s$, $r\geq2$ and $\sum_{s\in \gf_q^*} s=0$.

{Case 3:} Let $b\neq 0,c \in \gf_q$, $b'=0, c' \in \gf_q$. Then $\bc_{b,c}\cdot \bc_{0,c'}=0$ for $r \geq 2$ by the similar method in Case 2.

{Case 4:} Let $b\neq0,c \in \gf_q^*$, $b'\neq0, c' \in \gf_q^*$. Then
\begin{eqnarray*}
\bc_{(b,c)}\cdot \bc_{(b',c')} &=& \sum_{x\in D}\left(\left(\tr_{q^m/q}\left(bx\right)+c\right)\left(\tr_{q^m/q}\left(b'x\right)+c'\right)\right) \\
&=& \sum_{x\in D}\left(\tr_{q^m/q}\left(bx\right)\tr_{q^m/q}\left(b'x\right)\right) +
c\sum_{x\in D}\tr_{q^m/q}\left(b'x\right) + c'\sum_{x\in D}\tr_{q^m/q}\left(bx\right)+\sum_{x\in D}cc'\\
&=&  \sum_{x\in D}\left(\tr_{q^m/q}\left(bx\right)\tr_{q^m/q}\left(b'x\right)\right),
\end{eqnarray*}
where the last equality holds as $r \geq 2$ and $\sum_{s \in \gf_{q}^*} N_s s=0$ by Case 2.
\begin{enumerate}[1)]
\item If $b$ and $b'$ are $\gf_q$-linearly dependent, then there exists an element $l\in \gf_q^*$ such that $b = lb'$. We have
\begin{eqnarray*}
\bc_{(b,c)}\cdot \bc_{(b',c')} &=&  \sum_{x\in D}\left(\tr_{q^m/q}\left(lb'x\right)\right)\left(\tr_{q^m/q}\left(b'x\right)\right) \\
&=& l\sum_{s\in \gf_q^*}N_s s^2 \\
&=&lN_s\sum_{s\in \gf_q^*} s^2 =0
\end{eqnarray*}
as $r \geq 2$, $q\mid N_s$ and $N_s$ is independent on $s$ by Lemma \ref{lem-Ns}.

\item If $b$ and $b'$ are $\gf_q$-linearly independent, then
\begin{eqnarray*}
\bc_{(b,c)}\cdot \bc_{(b',c')} = \sum_{s,t\in \gf_q^*}N_{s, t}st,
\end{eqnarray*}
where $N_{s,t}$ was defined as in Lemma \ref{lem-Nst}.

\begin{itemize}
\item[2.1)] If $\tr_{q^r/q}\left(b'^{q^r+1}\right)=0$ and $\tr_{q^r/q}\left(bb'^{q^r}+b'b^{q^r}\right)=0,$ then by Lemma \ref{lem-Nst}, we have
\begin{eqnarray*}
\bc_{(b,c)}\cdot \bc_{(b',c')} &=&  q^{2r-3}\sum_{s,t\in \gf_q^*}st =0
\end{eqnarray*}
as $r \geq 2$.

\item[2.2)] If $\tr_{q^r/q}\left(b'^{q^r+1}\right)=0, \tr_{q^r/q}\left(b^{q^r+1}\right)=0$ and $\tr_{q^r/q}\left(bb'^{q^r}+b'b^{q^r}\right)\neq 0,$ then by Lemma \ref{lem-Nst} we have
\begin{eqnarray*}
\bc_{(b,c)}\cdot \bc_{(b',c')} &=&  q^{r-2}\left(q^{r-1}+1\right)\sum_{s,t\in \gf_q^*}st =0
\end{eqnarray*}
as $(q,r)\neq (2,2)$ and $\sum_{s,t \in \gf_q^*}st=0$.

\item[2.3)] If $\tr_{q^r/q}\left(b'^{q^r+1}\right)=0, \tr_{q^r/q}\left(b^{q^r+1}\right)\neq 0$ and $\tr_{q^r/q}\left(bb'^{q^r}+b'b^{q^r}\right)\neq 0,$ then by the third and fourth rows of Table \ref{tab1} in Lemma \ref{lem-Nst}, the value of  $N_{s,t}$ is dependent on the conditions
    $$t\tr_{q^r/q}\left(b^{q^r+1}\right)+ s\tr_{q^r/q}\left(bb'^{q^r}+b'b^{q^r}\right)=0$$ and $$t\tr_{q^r/q}\left(b^{q^r+1}\right) + s\tr_{q^r/q}\left(bb'^{q^r}+b'b^{q^r}\right) \neq 0.$$
    For fixed $t\in \gf_q^*$, let $s_t$ be the unique solution of the equation $t\tr_{q^r/q}(b^{q^r+1})+ s\tr_{q^r/q}(bb'^{q^r}+b'b^{q^r})=0$ with variable $s$. By Lemma \ref{lem-Nst}, we have
\begin{eqnarray*}
\bc_{(b,c)}\cdot \bc_{(b',c')} &=& \sum_{s,t\in \gf_q^*}N_{s,t}st \\
&=& \sum_{t\in \gf_q^*}\left(q^{r-1}\left(q^{r-2}-1\right)s_t+q^{2r-3}\sum_{s\in \gf_q^*\backslash\{s_t\}}s\right)t\\
&=& \sum_{t\in \gf_q^*}\left(q^{r-1}\left(q^{r-2}-1\right)s_t + q^{2r-3}(0-s_t)\right)t\\
&=& -q^{r-1}\sum_{t\in \gf_q^*}s_tt=0
\end{eqnarray*}
as $r \geq 2$.

\item[2.4)] If $\tr_{q^r/q}\left(b'^{q^r+1}\right)\neq0, \tr_{q^r/q}\left(b^{q^r+1}\right)= 0$ and $\tr_{q^r/q}\left(bb'^{q^r}+b'b^{q^r}\right)= 0,$ then by Lemma \ref{lem-Nst} we have
\begin{eqnarray*}
\bc_{(b,c)}\cdot \bc_{(b',c')} &=&  q^{2r-3}\sum_{s,t\in \gf_q^*}st
= 0
\end{eqnarray*}
as $r \geq 2$.

\item[2.5)] If $\tr_{q^r/q}\left(b'^{q^r+1}\right)\neq0, \tr_{q^r/q}\left(b^{q^r+1}\right)= 0$ and $\tr_{q^r/q}\left(bb'^{q^r}+b'b^{q^r}\right)\neq 0,$ then by the sixth and seventh rows of Table \ref{tab1} in Lemma \ref{lem-Nst}, the value of  $N_{s,t}$ is dependent on the conditions $s\tr_{q^r/q}\left(b'^{q^r+1}\right) + t\tr_{q^r/q}\left(bb'^{q^r}+b'b^{q^r}\right)=0$ and $s\tr_{q^r/q}\left(b'^{q^r+1}\right) + t\tr_{q^r/q}\left(bb'^{q^r}+b'b^{q^r}\right)\neq0$.
    Similarly to Case 2.3), we can also derive $\bc_{(b,c)}\cdot \bc_{(b',c')}=0$.

\item[2.6)] If $\tr_{q^r/q}\left(b'^{q^r+1}\right)\neq0, \tr_{q^r/q}\left(b^{q^r+1}\right)\neq 0$ and $\tr_{q^r/q}\left(bb'^{q^r}+b'b^{q^r}\right)= 0,$ then by the eighth and ninth rows of Table \ref{tab1} in Lemma \ref{lem-Nst}, the value of  $N_{s,t}$ is dependent on the conditions $s^2\tr_{q^r/q}\left(b'^{q^r+1}\right) = t^2\tr_{q^r/q}\left(b^{q^r+1}\right)$ and $s^2\tr_{q^r/q}\left(b'^{q^r+1}\right) \neq t^2\tr_{q^r/q}\left(b^{q^r+1}\right)$. For fixed $t\in \gf_q^*$, the equation $s^2\tr_{q^r/q}\left(b'^{q^r+1}\right) = t^2\tr_{q^r/q}\left(b^{q^r+1}\right)$ with variable $s$ has unique solution by Lemma \ref{lem-root} which is denoted by $s_t$. Furthermore, $s_t=\frac{\tr_{q^r/q}\left(b^{q^r+1}\right)^{\frac{q}{2}}t}{\tr_{q^r/q}\left(b'^{q^r+1}\right)^{\frac{q}{2}}}$. Then by Lemma \ref{lem-Nst}, we have
    \begin{eqnarray*}
& &\bc_{(b,c)}\cdot \bc_{(b',c')} \\
&=& \sum_{s,t\in \gf_q^*}N_{s,t}st \\
&=& \sum_{t\in \gf_q^*}\left(q^{r-2}(q^{r-1}-1)s_t + q^{2r-3}\sum_{s\in \gf_q^*\backslash \{s_t\}}s\right)t \\
&=&-q^{r-2}\sum_{t\in \gf_q^*}s_tt\\
&=&-q^{r-2}\frac{\tr_{q^r/q}\left(b^{q^r+1}\right)^{\frac{q}{2}}}{\tr_{q^r/q}\left(b'^{q^r+1}\right)^{\frac{q}{2}}}\sum_{t\in \gf_q^*}t^2=0
    \end{eqnarray*}
    as  $(q,r) \neq (2,2)$.

\item[2.7)] If $\tr_{q^r/q}\left(b'^{q^r+1}\right)\neq0$, $\tr_{q^r/q}\left(b^{q^r+1}\right)\neq0$ and $\tr_{q^r/q}\left(bb'^{q^r}+b'b^{q^r}\right)\neq 0$, by the tenth and eleventh rows of Table \ref{tab1} in  Lemma \ref{lem-Nst}, the value of  $N_{s,t}$ is dependent on the conditions
    $$t^2\tr_{q^r/q}\left(b^{q^r+1}\right)+st\tr_{q^r/q}\left(bb'^{q^r}+b'b^{q^r}\right)+s^2\tr_{q^r/q}\left(b'^{q^r+1}\right)=0$$
    and
        $$t^2\tr_{q^r/q}\left(b^{q^r+1}\right)+st\tr_{q^r/q}\left(bb'^{q^r}+b'b^{q^r}\right)+s^2\tr_{q^r/q}\left(b'^{q^r+1}\right)\neq 0.$$
Similarly to Case 2.6), we also have $\bc_{(b,c)}\cdot \bc_{(b',c')}=0$.
\end{itemize}
\end{enumerate}

Similarly to Case 4, we can also verify that $\bc_{(b,0)}\cdot \bc_{(b',0)}=0$, $\bc_{(b,c)}\cdot \bc_{(b',0)}=0$ and $\bc_{(b,0)}\cdot \bc_{(b',c')}=0$.
Then we conclude that $\overline{\cC_{D}}$ is also self-orthogonal for even $q$.
\end{proof}

Note that the locality of $\overline{\cC_D}$ was determined in Theorem \ref{loc} for the case $q > 2$. In the following, we will determine the locality of $\overline{\cC_D}$ for the case $q=2$ if $D$ is given in Equation $(\ref{eq-Dr})$. Here we first give a corollary which shows that $\overline{\cC_D}$ in Theorem \ref{tem-wtCD1} holds 2-designs
for $q=2$.

\begin{corollary}\label{col-design1}
  Let $q=2$ and $r \geq 3$. Then the linear code $\overline{\cC_{D}}$ with $D$ in Equation $(\ref{eq-Dr})$
   has parameters $[2^{r-1}(2^r-1), m+1, 2^{r-1}(2^{r-1}-1)]$ and weight enumerator
  \begin{eqnarray*}
  A(z)=1+(2^{2r}-1)z^{2^{r-1}(2^{r-1}-1)}+(2^{2r}-1)z^{2^{2r-2}}+z^{2^{r-1}(2^r-1)}.
  \end{eqnarray*}
  The dual code $\overline{\cC_{D}}^{\perp}$ has parameters $[2^{r-1}(2^r-1),2^{r-1}(2^r-1)-m-1, 4]$. Then the minimum weight codewords in $\overline{\cC_D}$ hold a $2$-$(2^{r-1}(2^r-1), 2^{r-1}(2^{r-1}-1), 2^{2r-2}-2^{r-1}-1)$ design. The minimum weight codewords in $\overline{\cC_D}^{\perp}$ hold a $2$-$(2^{r-1}(2^r-1), 4, 2^{2r-3}-2^{r-2}-1)$ design.
\end{corollary}
\begin{proof}
  The desired conclusions follow from Theorem \ref{tem-wtCD1} and the Assmus-Mattson Theorem \cite{H}.
\end{proof}

Then we have the following result.
\begin{theorem}\label{th-loc-q=2}
  Let $q=2$ and $m=2r$ with $r \geq 3$. Then $\overline{\cC_D}$ with $D$ in Equation $(\ref{eq-Dr})$ is a
  $$(2^{r-1}(2^r-1), m+1, 2^{r-1}(2^{r-1}-1), 2; 3)\text{- LRC}$$ and $\overline{\cC_D}^{\perp}$ is a $$(2^{r-1}(2^r-1),2^{r-1}(2^r-1)-m-1, 4, 2; 2^{r-1}(2^{r-1}-1)-1)\text{- LRC}.$$ In particular, $\overline{\cC_D}^{\perp}$ is $k$-optimal or almost $k$-optimal.
\end{theorem}
\begin{proof}
By Corollary \ref{col-design1} and Lemma \ref{lem-locality}, we derive that the locality of $\overline{\cC_D}$ is $3$ and the locality of $\overline{\cC_D}^{\perp}$ is $2^{r-1}(2^{r-1}-1)-1$. Now we prove $\overline{\cC_D}^{\perp}$ is an optimal or almost optimal LRC.
By Lemma \ref{lem-CMbound}, putting the parameters of the $(2^{r-1}(2^r-1),2^{r-1}(2^r-1)-2r-1, 4, 2; 2^{r-1}(2^{r-1}-1)-1)-$ LRC into the right-hand side of the Cadambe-Mazumdar bound in (\ref{eqn-CMbound}), we have
\begin{eqnarray*}
  k &\leq& \mathop{\min}_{t \in \mathbb{Z}^{+}} [rt+k_{opt}^{(q)}(n-t(r+1),d)]\\
  &\leq& \mathop{\min}_{t=1}[rt+k_{opt}^{(q)}(n-t(r+1),d)]\\
  &=& r+k_{opt}^{(q)}(n-(r+1),d) \\
   &=& 2^{r-1}(2^{r-1}-1)-1+k_{opt}^{(2)}(2^{2r-2},4).
\end{eqnarray*}
Thanks to the sphere-packing bound, we obtain
\begin{eqnarray*}
  k_{opt}^{(2)}(2^{2r-2},4) &\leq& \left\lfloor \log_2 \frac{2^{2^{2r-2}}}{\sum_{i=0}^{1}\tbinom{n}{i}}  \right\rfloor\\
  &=& \left\lfloor \log_2 \frac{2^{2^{2r-2}}}{1+2^{2r-2}}  \right\rfloor\\
  &=& \left\lfloor \log_2 (2^{2^{2r-2}}) - \log_2 (1+2^{2r-2})  \right\rfloor\\
  &=& 2^{2r-2}-2r+1,
\end{eqnarray*}
where the last equality holds as $2^{2r-2} < 1+2^{2r-2} < 2^{2r-1}$ by the sphere-packing bound. Thus, $$2^{r-1}(2^r-1)-2r-1=k \leq \mathop{\min}_{t \in \mathbb{Z}^{+}} [rt+k_{opt}^{(q)}(n-t(r+1),d)] \leq 2^{r-1}(2^r-1)-2r.$$
Then $$\mathop{\min}_{t \in \mathbb{Z}^{+}} [rt+k_{opt}^{(q)}(n-t(r+1),d)]=2^{r-1}(2^r-1)-2r \mbox{ or } 2^{r-1}(2^r-1)-2r-1$$ and $\overline{\cC_D}^{\perp}$ is optimal or almost optimal with respect to the Cadambe-Mazumdar bound.
\end{proof}

In the following, we will prove that $\overline{\cC_{D}}$ is an almost optimally or optimally extendable code if we choose a suitable generator matrix.
Let $I_{n,n}$ denote the identity of size $n$. Let $\gf_{q^m}^*=\langle\alpha\rangle$. Then $\{1,\alpha, \alpha^{2},\ldots, \alpha^{m-1}\}$ is a $\gf_q$-basis of $\gf_{q^m}$.
For the linear code $\overline{\cC_{D}}$ in Theorem \ref{tem-wtCD1}, it has a generator matrix $G$ by the proof of Theorem \ref{loc}.
Now we can obtain another generator matrix $G_1$ of $\overline{\cC_{D}}$ by an elementary row operation  on $G$, where
\begin{eqnarray}\label{matrix1}
G_1=\left[
\begin{array}{cccc}
1&1&\cdots&1 \\
 \tr_{q^m/q}(\alpha^{0}d_{1})& \tr_{q^m/q}(\alpha^{0}d_{2})& \cdots &\tr_{q^m/q}(\alpha^{0}d_{n}) \\
 \tr_{q^m/q}(\alpha^{1}d_{1})+1& \tr_{q^m/q}(\alpha^{1}d_{2})+1& \cdots &\tr_{q^m/q}(\alpha^{1}d_{n})+1 \\
 \vdots &\vdots &\ddots &\vdots \\
\tr_{q^m/q}(\alpha^{m-1}d_{1})& \tr_{q^m/q}(\alpha^{m-1}d_{2})& \cdots &\tr_{q^m/q}(\alpha^{m-1}d_{n}) \\
\end{array}\right],
\end{eqnarray}
where $D=\{d_1,d_2,\cdots,d_n\}$ is defined in Equation (\ref{eq-Dr}).

\begin{theorem}\label{th-LCD1}
Let $q = p^e > 2$, where $p$ is a prime and $e$ is a positive integer. Let $m, r$ be positive integers such that $m=2r$ and $r \geq 2$. Let $G_1$ be defined as above and $G_1'=[I_{m+1, m+1}: G_1]$. Then the linear code $\overline{\cC_{D}}'$ generated by matrix $G_1'$ is a linear code with parameters $$[q^{r-1}(q^r-q+1)+m+1, m+1, d'\geq q^{r-1}(q^{r}-q^{r-1}-q+1)+1].$$ Besides, $\overline{\cC_{D}}'^{\perp}$ has parameters $[q^{r-1}(q^r-q+1)+m+1, q^{r-1}(q^r-q+1), 2 \leq d'^{\perp} \leq 3]$.
\end{theorem}
\begin{proof}
 Obviously, the length of $\overline{\cC_{D}}'$ is $m+1+q^{r-1}(q^r-q+1)$ and the dimension of $\overline{\cC_{D}}'$ is $m+1$. Let $d$, $d'$, $d'^{\perp}$ respectively represent the minimum distance of $\overline{\cC_{D}}$, $\overline{\cC_{D}}'$ and $\overline{\cC_{D}}'^{\perp}$. It is clear that the minimum distance $d' \geq d$, where $d=q^{r-1}(q^{r}-q^{r-1}-q+1)$. For a codeword $\ba'$ in $\overline{\cC_{D}}'$, we have
\begin{eqnarray*}
\ba'&=&(l,a_0,a_1 \cdots, a_{m-1})G_1'\\&=&(l, a_0,a_1, \cdots, a_{m-1}, \sum_{i=0}^{m-1}a_i\tr_{q^m/q}(\alpha^{i}d_{1})+l+a_1, \cdots, \sum_{i=0}^{m-1}a_i\tr_{q^m/q}(\alpha^{i}d_{n}))+l+a_1),
\end{eqnarray*}
where $l, a_i \in \gf_q$ for all $0 \leq i \leq m-1$.  Assume that $\ba'$ has the minimum distance $d'$. If $d'=d$, then the first $m+1$ locations of $\ba'$ are all zero. This means that $l=a_i=0$ $(0 \leq i \leq m-1)$ and $\ba'=\mathbf{0}$, which contradicts with $\text{wt}(\ba')=d'=d$. Then $d' \geq d+1$. Besides, since any column in $G_1'$ is nonzero and the minimum distance of $\overline{\cC_{D}}^\perp$ is $3$, we have $2 \leq d'^{\perp} \leq 3$. The proof is completed.
\end{proof}

\begin{corollary}\label{col-LCD1}
Let $q = p^e$, where $p$ is a prime and $e$ is a positive integer. Let $m, r$ be positive integers such that $m=2r$ and $r \geq 2$. Let $G_1$ be defined as above and $G_1'=[I_{m+1, m+1}: G_1]$. If $p \nmid r$, then the linear code $\overline{\cC_{D}}'$ generated by $G_1'$ has parameters  $[q^{r-1}(q^r-q+1)+m+1, m+1, q^{r-1}(q^{r}-q^{r-1}-q+1)+1].$
\end{corollary}
\begin{proof}
By Theorem \ref{th-LCD1}, we obtain $d' \geq q^{r-1}(q^{r}-q^{r-1}-q+1)+1$. In the following, we will prove that $d' = q^{r-1}(q^{r}-q^{r-1}-q+1)+1$. Let $a_0 \in \gf_q^*$ and $l=a_1=a_2=a_3=\cdots=a_{m-1}=0$. Then $\ba'=(0,a_0,0, \cdots, 0, a_0\tr_{q^m/q}(d_1), \cdots, a_0\tr_{q^m/q}(d_n))\in \overline{\cC_{D}}'$. By Theorem \ref{tem-wtCD1}, we deduce that $\text{wt}(\ba')=q^{r-1}(q^{r}-q^{r-1}-q+1)+1$ as $\tr_{q^r/q}(a_0^{q^r+1})=\tr_{q^r/q}(a_0^{2})=ra_0^2 \neq 0$. Then there exist codewords in $\overline{\cC_{D}}'$ with Hamming weight $q^{r-1}(q^{r}-q^{r-1}-q+1)+1$. The desired conclusion follows.
\end{proof}

\begin{remark}\label{rem1}
 Note that the linear code generated by the matrix $[I_{m+1, m+1}: G]$ is dependent on the generator matrix $G$ of $\overline{\cC_{D}}$. By Magma, we find that
 different $G$ may yield inequivalent codes.
 For example, if we select
\begin{eqnarray*}
G=\left[
\begin{array}{cccc}
1&1&\cdots&1 \\
\tr_{q^m/q}(\alpha^{0}d_{1})& \tr_{q^m/q}(\alpha^{0}d_{2})& \cdots &\tr_{q^m/q}(\alpha^{0}d_{n}) \\
\tr_{q^m/q}(\alpha^{1}d_{1})& \tr_{q^m/q}(\alpha^{1}d_{2})& \cdots &\tr_{q^m/q}(\alpha^{1}d_{n}) \\
 \vdots &\vdots & &\vdots \\
\tr_{q^m/q}(\alpha^{m-1}d_{1})& \tr_{q^m/q}(\alpha^{m-1}d_{2})& \cdots &\tr_{q^m/q}(\alpha^{m-1}d_{n}) \\
\end{array}\right],
\end{eqnarray*}
then $\widetilde{\overline{\cC_{D}}'}$ generated by $[I_{m+1,m+1}:G]$ has different weight distribution from $\overline{\cC_{D}}'$ by Magma. Particularly, the minimum distance of $\widetilde{\overline{\cC_{D}}'}^\perp$ is $2$, while the minimum distance of $\overline{\cC_{D}}'^{\perp}$ is $3$ in most cases. It is interesting to find suitable generator matrix $G$ such that the linear code generated by $[I: G]$ and its dual have better parameters.
\end{remark}

The following theorem presents a family of (almost) optimally extendable self-orthogonal codes.

\begin{theorem}\label{extend1}
  Let $\overline{\cC_{D}}$ be the self-orthogonal code in Theorem \ref{tem-wtCD1} and $\overline{\cC_{D}}'$ be defined as in Theorem \ref{th-LCD1}. If $q > 2$ and $r \geq 2$, then $\overline{\cC_{D}}$ is an almost optimally or optimally extendable self-orthogonal code.
\end{theorem}

\begin{proof}
By Theorems \ref{tem-wtCD1} and \ref{th-LCD1}, we have $0 \leq d(\overline{\cC_{D}}^{\perp})-d(\overline{\cC_{D}}'^{\perp}) \leq 1$. Then $\overline{\cC_{D}}$ is an almost optimally or optimally extendable self-orthogonal code.
\end{proof}

By Magma, we find that $\overline{\cC_{D}}'^\perp $ has minimum distance $3$ in most cases. We give two examples in the following.
\begin{example}\label{exa-1}
Let $q=3, r=2$. By Magma, then the ternary linear codes $\overline{\cC_{D}}$, $\overline{\cC_{D}}'$, $\overline{\cC_{D}}^{\perp}$, $\overline{\cC_{D}}'^{\perp}$ have parameters $$[21, 5, 12], [26,5,13], [21, 16, 3], [26, 21, 3],$$ respectively.
Besides, $\overline{\cC_{D}}$ is an optimally extendable self-orthogonal code, and all of $\overline{\cC_{D}}$, $\overline{\cC_{D}}^{\perp}$ and $\overline{\cC_{D}}'^{\perp}$ are optimal codes by the Code Tables in \cite{Codetable}.
\end{example}

\begin{example}\label{exa-2}
Let $q=3, r=3$. By Magma, then the ternary linear codes $\overline{\cC_{D}}$, $\overline{\cC_{D}}'$, $\overline{\cC_{D}}^{\perp}$, $\overline{\cC_{D}}'^{\perp}$ have parameters $$[225, 7, 144], [232, 7, 145], [225, 218, 3], [232, 225, 3],$$ respectively.
Besides, $\overline{\cC_{D}}$ is an optimally extendable self-orthogonal code, and all of $\overline{\cC_{D}}$, $\overline{\cC_{D}}^{\perp}$ and $\overline{\cC_{D}}'^{\perp}$ are optimal codes by the Code Tables in \cite{Codetable}.
\end{example}

\subsection{The case that $N=2$}
In this subsection, let $q$ be an odd prime power, $m$ be a positive integer such that $m \geq 3$. Let $N=2$ and the defining set be
\begin{eqnarray}\label{eq-D2}
  D=\{x \in \gf_{q^m}: \tr_{q^m/q}(x^{2})=0\}.
\end{eqnarray}
We first prove that $\overline{\cC_D}$ is self-orthogonal and determine its parameters and weight distribution. Then we will prove that $\overline{\cC_D}$ is also (almost) optimally extendable. By Magma, if $m$ is even, we find that the parameters and weight distribution of $\overline{\cC_D}$ are the same whether $N=2$ or $N=q^r+1$ in most cases. Hence we only consider the case for odd $m$ in this subsection.

\begin{lemma}\label{lem-Ns1}
Let $s \in \gf_q^*$ and $b \in \gf_{q^m}^*$, where $m$ is an odd integer with $m \geq 3$ and $q=p^e$ with odd prime $p$. Let $N_s$ denote the number of solutions in $\gf_{q^m}$ of the following system of equations as
\begin{eqnarray*}\label{eq-01}
\left\{
\begin{array}{l}
\tr_{q^m/q}\left(bx\right)=s, \\
\tr_{q^r/q}(x^{2})=0.
\end{array}
\right.
\end{eqnarray*}
Then
\begin{eqnarray*}
N_s=
\begin{cases}
q^{m-2}     &\text{if $\tr_{q^r/q}\left(b^{2}\right)=0$,}\\
q^{m-2}-(-1)^{l}\eta(\tr_{q^m/q}(b^2))q^{\frac{m-3}{2}}     &\text{if $\tr_{q^r/q}\left(b^{2}\right)\neq0$,}
\end{cases}
\end{eqnarray*}
where $l=\frac{e(p-1)(m+3)}{4}.$
\end{lemma}

\begin{proof}
By the orthogonal relation of additive characters, we have
\begin{eqnarray}\label{eq-Ns1}
N_s &=&| \{x\in\gf_{q^m}: \tr_{q^m/q}\left(bx\right)=s, \tr_{q^m/q}(x^{2})=0\}|\nonumber\\
&=& \frac{1}{q^2}\sum_{u\in\gf_q}\sum_{v\in\gf_q}\sum_{x\in\gf_{q^m}}\zeta_p^{\tr_{q/p}\left(\tr_{q^m/q}\left(ubx\right)-us\right)}
\zeta_p^{\tr_{q/p}\left(\tr_{q^m/q}\left(vx^{2}\right)\right)}\nonumber\\
&=&\frac{1}{q^2}\left(\sum_{u\in\gf_q^*}\sum_{x\in\gf_{q^m}}\chi_1\left(ubx\right)\phi_1\left(-us\right) +
\sum_{v\in\gf_q^*}\sum_{x\in\gf_{q^m}}\chi_1(vx^2) + \right.\nonumber\\
&&\left.\sum_{v\in\gf_q^*}\sum_{u\in\gf_q^*}\phi_1(-us)\sum_{x\in\gf_{q^m}}\chi_1(vx^2+ubx)+ q^m\right).
\end{eqnarray}
By the orthogonal relation of additive characters, we have
\begin{eqnarray}\label{eq-S1}
S_1:=\sum_{u\in\gf_q^*}\sum_{x\in\gf_{q^m}}\chi_1\left(ubx\right)\phi_1\left(-us\right)
=\sum_{u\in\gf_q^*}\phi_1\left(-us\right)\sum_{x\in\gf_{q^m}}\chi_1\left(ubx\right)=0.
\end{eqnarray}
By Lemma \ref{lem-weil}
\begin{eqnarray}\label{eq-S2}
S_2:=\sum_{v\in\gf_q^*}\sum_{x\in\gf_{q^m}}\chi_1(vx^2)
=G(\eta' , \chi_1)\sum_{v\in\gf_q^*}\eta'(v)
=G(\eta' , \chi_1)\sum_{v\in\gf_q^*}\eta(v)=0,
\end{eqnarray}
and
\begin{eqnarray}\label{eq-S3}
S_3&:=&\sum_{v\in\gf_q^*}\sum_{u\in\gf_q^*}\phi_1(-us)\sum_{x\in\gf_{q^m}}\chi_1(vx^2+ubx)\nonumber\\
&=&\sum_{u\in\gf_q^*}\phi_1(-us)\sum_{v\in\gf_q^*}\chi_1\left(-\frac{u^2b^2}{4v}\right)\eta'(v)G(\eta', \chi_1)\nonumber\\
&=&G(\eta', \chi_1)\sum_{u\in\gf_q^*}\phi_1(-us)\sum_{v\in\gf_q^*}\phi_1\left(-\frac{u^2\tr_{q^m/q}(b^2)}{4v}\right)\eta(v)\nonumber\\
&=&\left\{\begin{array}{lll}
G(\eta', \chi_1)\sum_{u\in\gf_q^*}\phi_1(-us)\sum_{v\in\gf_q^*}\eta(v) \qquad\qquad\qquad\qquad\qquad\qquad  \mbox{if $\tr_{q^r/q}\left(b^{2}\right)=0$}\\
G(\eta', \chi_1)\sum_{u\in\gf_q^*}\phi_1(-us)\sum_{v\in\gf_q^*}\phi_1\left(-\frac{u^2\tr_{q^m/q}(b^2)}{4v}\right)\eta\left(-\frac{u^2\tr_{q^m/q}(b^2)}{4v}\right)\eta(-\tr_{q^m/q}(b^2))  \\   \qquad\qquad\qquad\qquad\qquad\qquad\qquad\qquad\qquad\qquad\qquad\qquad\qquad\quad     \mbox{if $\tr_{q^r/q}\left(b^{2}\right)\neq0$}
\end{array} \right. \nonumber\\
&=&\begin{cases}
     0 & \mbox{if $\tr_{q^r/q}\left(b^{2}\right)=0$} \\
     -G(\eta',\chi_1)G(\eta, \phi_1)\eta(-\tr_{q^m/q}(b^2)) & \mbox{if $\tr_{q^r/q}\left(b^{2}\right)\neq0$}
   \end{cases}\nonumber\\
&=&\begin{cases}
     0 & \mbox{if $\tr_{q^r/q}\left(b^{2}\right)=0$,} \\
     -(-1)^l\eta(\tr_{q^m/q}(b^2))q^{\frac{m+1}{2}} & \mbox{if $\tr_{q^r/q}\left(b^{2}\right)\neq0$},
   \end{cases}
\end{eqnarray}
where the third equality holds as $\eta'(v)=\eta(v)$ for $v \in \gf_q$ and odd $m$, and the last equality holds as
\begin{eqnarray*}
  G(\eta',\chi_1)G(\eta, \phi_1)\eta(-1)=(-1)^lq^{\frac{m+1}{2}}
\end{eqnarray*}
by Lemma \ref{quadGuasssum1}.

Substituting Equations (\ref{eq-S1}), (\ref{eq-S2}) and (\ref{eq-S3}) into Equation (\ref{eq-Ns1}) yields the desired conclusion.
\end{proof}

\begin{theorem}\label{tem-wtCDN2}
Let $q = p^e$, where $p$ is an odd prime and $e$ is a positive integer. Let $m$ be an odd integer with $m \geq 3$. Denote by $l=\frac{e(p-1)(m+3)}{4}.$ Then the augmented code $\overline{\cC_{D}}$ with the defining set $D=\{x \in \gf_{q^m} : \tr_{q^m/q}(x^2)=0\}$ has parameters $[q^{m-1}, m+1, (q-1)q^{\frac{m-3}{2}}(q^{\frac{m-1}{2}}-1)]$ and its weight distribution is listed in Table \ref{tab-2}. In particular, $\overline{\cC_{D}}$ is a family of self-orthogonal codes for $m > 3$.  Besides, $\overline{\cC_{D}}^{\perp}$ has parameters $[q^{m-1}, q^{m-1}-m-1, 3]$ and is at least almost optimal according to the sphere-packing bound.
\begin{table}[h!]
\begin{center}
\caption{The weight distribution of $\overline{\cC_{D}}$ in Theorem \ref{tem-wtCDN2}.}\label{tab-2}
\begin{tabular}{@{}ll@{}}
\toprule
Weight & Frequency  \\
\midrule
$0$ & $1$ \\
$q^{m-1}$ &  $q-1$ \\
$(q-1)q^{m-2}$ &  $q(q^{m-1}-1)$ \\
$(q-1)q^{m-2}+(-1)^{l}q^{\frac{m-3}{2}}$ &  $\frac{(q-1)^2(q^{m-1}+(-1)^{l}q^{\frac{m-1}{2}})}{2}$ \\
$(q-1)q^{m-2}-(-1)^{l}q^{\frac{m-3}{2}}$ &  $\frac{(q-1)^2(q^{m-1}+(-1)^{l+1}q^{\frac{m-1}{2}})}{2}$ \\
$(q-1)q^{m-2}-(-1)^{l}(q-1)q^{\frac{m-3}{2}}$ & $\frac{(q-1)(q^{m-1}+(-1)^{l}q^{\frac{m-1}{2}})}{2}$ \\
$(q-1)q^{m-2}+(-1)^{l}(q-1)q^{\frac{m-3}{2}}$ & $\frac{(q-1)(q^{m-1}+(-1)^{l+1}q^{\frac{m-1}{2}})}{2}$\\
\bottomrule
\end{tabular}
\end{center}
\end{table}
\end{theorem}

\begin{proof}
For $b=0$, it is obvious that the codeword $\bc_{(b,c)}$ has Hamming weight
\begin{eqnarray*}
\text{wt}(\bc_{(b,c)})=
\begin{cases}
0 &\text{if $c=0$,}\\
q^{m-1}     &\text{if $c \neq 0$.}
\end{cases}
\end{eqnarray*}
For $b \in \gf_q^*$, we denote by $N_{(b,c)}=|\{x \in D:\tr_{q^m/q}(bx)+c=0\}|=|\{x \in \gf_{q^m}:\tr_{q^m/q}(x^2)=0 \mbox{ and } \tr_{q^m/q}(bx)+c=0\}|.$ Then the Hamming weight of codeword $\bc_{(b,c)}$ is $\text{wt}(\bc_{(b,c)})=n-N_{(b,c)}$ for $b \in \gf_q^*, c \in \gf_q$. For $c \in \gf_q^*$, we can find that the value of $N_{(b,c)}$ is equal to the value of $N_s$ in Lemma \ref{lem-Ns1} for $s=-c$. Then by Lemma \ref{lem-Ns1}, we have
\begin{eqnarray*}
\text{wt}(\bc_{(b,c)})=n-N_{(b,c)}
=\begin{cases}
q^{m-1}-q^{m-2}     &\text{if $\tr_{q^m/q}\left(b^{2}\right)=0$,}\\
q^{m-1}-q^{m-2}+(-1)^{l}\eta(\tr_{q^m/q}(b^2))q^{\frac{m-3}{2}}     &\text{if $\tr_{q^m/q}\left(b^{2}\right)\neq0$.}
\end{cases}
\end{eqnarray*}
For $c=0$, similarly to the proof of Lemma \ref{lem-Ns1}, we can also derive that
\begin{eqnarray*}
\text{wt}(\bc_{(b,c)})&=&\begin{cases}
q^{m-1}-q^{m-2} &\text{if $\tr_{q^m/q}\left(b^{2}\right)=0$,}\\
q^{m-1}-q^{m-2}-(-1)^{l}\eta(\tr_{q^m/q}(b^2))(q-1)q^{\frac{m-3}{2}}    &\text{if $\tr_{q^m/q}\left(b^{2}\right)\neq0$.}
\end{cases}
\end{eqnarray*}

Summarizing the above discussions, we derive that
\begin{eqnarray*}
\text{wt}(\bc_{(b,c)})&=&\begin{cases}
0 & \mbox{if $b=0$ and $c=0$,}\\
q^{m-1} & \mbox{if $b=0$ and $c \neq 0$,}\\
(q-1)q^{m-2} & \mbox{if $b \in \gf_{q^m}^*, \tr_{q^m/q}\left(b^{2}\right)=0$ and $c \in \gf_q$,}\\
(q-1)q^{m-2}+(-1)^{l}q^{\frac{m-3}{2}}   & \mbox{if $b \in \gf_{q^m}^*, \eta(\tr_{q^m/q}(b^2))=1$ and $c\neq0$,}\\
(q-1)q^{m-2}-(-1)^{l}q^{\frac{m-3}{2}}   & \mbox{if $b \in \gf_{q^m}^*, \eta(\tr_{q^m/q}(b^2))=-1$ and $c\neq0$,}\\
(q-1)q^{m-2}-(-1)^{l}(q-1)q^{\frac{m-3}{2}}   & \mbox{if $b \in \gf_{q^m}^*, \eta(\tr_{q^m/q}(b^2))=1$ and $c=0$,}\\
(q-1)q^{m-2}+(-1)^{l}(q-1)q^{\frac{m-3}{2}}   & \mbox{if $b \in \gf_{q^m}^*, \eta(\tr_{q^m/q}(b^2))=-1$ and $c=0$.}
\end{cases}
\end{eqnarray*}

Denote by $w_1=q^{m-1}$, $w_2=(q-1)q^{m-2}$, $w_3=(q-1)q^{m-2}+(-1)^{l}q^{\frac{m-3}{2}}$, $w_4=(q-1)q^{m-2}-(-1)^{l}q^{\frac{m-3}{2}}$, $w_5=(q-1)q^{m-2}-(-1)^{l}(q-1)q^{\frac{m-3}{2}}$, $w_6=(q-1)q^{m-2}+(-1)^{l}(q-1)q^{\frac{m-3}{2}}$. Next we will determine the frequency $A_{w_i}$, $1 \leq i \leq 6$. It is obvious that $A_{w_1}=q-1$ and $A_{w_2}=q(q^{m-1}-1)$. Denote by $A_a=|\{b \in \gf_{q^m}^*:\tr_{q^m/q}(b^2)=a\}|$, where $a \in \gf_q$.
It is easy to prove that $A_a=q^{m-1}+(-1)^{l}\eta(a)q^{\frac{m-1}{2}}$ by the orthogonal relation of additive characters and quadratic Gaussian sums. Then $A_{w_3}=\frac{(q-1)^2(q^{m-1}+(-1)^{l}q^{\frac{m-1}{2}})}{2}$, $A_{w_4}=\frac{(q-1)^2(q^{m-1}+(-1)^{l+1}q^{\frac{m-1}{2}})}{2}$, $A_{w_5}=\frac{(q-1)(q^{m-1}+(-1)^{l}q^{\frac{m-1}{2}})}{2}$ and $A_{w_6}=\frac{(q-1)(q^{m-1}+(-1)^{l+1}q^{\frac{m-1}{2}})}{2}$. Besides, by Theorem \ref{loc}, $d(\overline{\cC_D}^{\perp})=3$ and then $\overline{\cC_D}^{\perp}$ has parameters $[q^{m-1}, q^{m-1}-m-1, 3]$. It is easy to verify that $\overline{\cC_D}^{\perp}$ is at least almost optimal according to the sphere-packing bound.

By the above discussion, we deduce that $\mathbf{1} \in \overline{\cC_D}$ and $\overline{\cC_D}$ is $p$-divisible for $m >3$. Then by Lemma \ref{lem-self-orthogonal}, $\overline{\cC_D}$ is a family of self-orthogonal codes for $m > 3$. The proof is completed.
\end{proof}

Let $I_{n,n}$ denote the identity of size $n$. Let $\gf_{q^m}^*=\langle\alpha\rangle$. Then $\{1,\alpha, \alpha^{2},\ldots, \alpha^{m-1}\}$ is a $\gf_q$-basis of $\gf_{q^m}$.
For the linear code $\overline{\cC_{D}}$ in Theorem \ref{tem-wtCDN2}, it has a generator matrix $G$ by the proof of Theorem \ref{loc}.
Now we can obtain another generator matrix $G_2$ of $\overline{\cC_{D}}$ by an elementary row operation  on $G$, where
\begin{eqnarray*}
G_2:=\left[
\begin{array}{cccc}
1&1&\cdots&1 \\
 \tr_{q^m/q}(\alpha^{0}d_{1})& \tr_{q^m/q}(\alpha^{0}d_{2})& \cdots &\tr_{q^m/q}(\alpha^{0}d_{n}) \\
 \tr_{q^m/q}(\alpha^{1}d_{1})+1& \tr_{q^m/q}(\alpha^{1}d_{2})+1& \cdots &\tr_{q^m/q}(\alpha^{1}d_{n})+1 \\
 \vdots &\vdots &\ddots &\vdots \\
\tr_{q^m/q}(\alpha^{m-1}d_{1})& \tr_{q^m/q}(\alpha^{m-1}d_{2})& \cdots &\tr_{q^m/q}(\alpha^{m-1}d_{n}) \\
\end{array}\right],
\end{eqnarray*}
and $D=\{d_1,d_2,\cdots,d_n\}$ is defined in Equation (\ref{eq-D2}).
\begin{theorem}\label{th-N2}
Let $q = p^e$, where $p$ is an odd prime and $e$ is a positive integer. Let $m$ be an odd integer with $m \geq 3$. Let $G_2$ be defined as above and $G_2'=[I_{m+1, m+1}: G_2]$. Then the linear code $\overline{\cC_{D}}'$ generated by matrix $G_2'$ is a linear code with parameters $$[q^{m-1}+m+1, m+1, d'\geq (q-1)q^{\frac{m-3}{2}}(q^{\frac{m-1}{2}}-1)+1].$$ Besides, $\overline{\cC_{D}}'^{\perp}$ has parameters $[q^{m-1}+m+1, q^{m-1}, 2 \leq d'^{\perp} \leq 3]$.
\end{theorem}

\begin{proof}
  Similarly to the proof of Theorem \ref{th-LCD1}, the desired conclusions follow.
\end{proof}

\begin{theorem}\label{extend2}
  Let $\overline{\cC_{D}}$ be the code in Theorem \ref{tem-wtCDN2} and $\overline{\cC_{D}}'$ be defined as in Theorem \ref{th-N2}. Then $\overline{\cC_{D}}$ is an almost optimally or optimally extendable code.
\end{theorem}

\begin{proof}
By Theorems \ref{tem-wtCDN2} and \ref{th-N2}, we have $0 \leq d(\overline{\cC_{D}}^{\perp})-d(\overline{\cC_{D}}'^{\perp}) \leq 1$. Then $\overline{\cC_{D}}$ is an almost optimally or optimally extendable self-orthogonal code.
\end{proof}

By Magma, we find that $\overline{\cC_{D}}'^\perp $ for $N=2$ has minimum distance $3$ in most cases. We give two examples in the following.

\begin{example}\label{exa-5}
Let $q=3, m=5$. By Magma, then the ternary linear codes $\overline{\cC_{D}}$, $\overline{\cC_{D}}'$, $\overline{\cC_{D}}^{\perp}$, $\overline{\cC_{D}}'^{\perp}$ have parameters $$[81, 6, 48], [87, 6, 49], [81, 75, 3], [87, 81, 3],$$ respectively.
Besides, $\overline{\cC_{D}}$ is an optimally extendable self-orthogonal code, and both of $\overline{\cC_{D}}^{\perp}$ and $\overline{\cC_{D}}'^{\perp}$ are optimal codes by the Code Tables in \cite{Codetable}.
\end{example}

\begin{example}\label{exa-6}
Let $q=5, m=3$. By Magma, then the ternary linear codes $\overline{\cC_{D}}$, $\overline{\cC_{D}}'$, $\overline{\cC_{D}}^{\perp}$, $\overline{\cC_{D}}'^{\perp}$ have parameters $$[25, 4, 16], [29, 4, 17], [25, 21, 3], [29, 25, 3],$$ respectively.
Besides, $\overline{\cC_{D}}$ is an optimally extendable self-orthogonal code, $\overline{\cC_{D}}^{\perp}$ is an almost optimal code by the Code Tables in \cite{Codetable},  and $\overline{\cC_{D}}'^{\perp}$ is an optimal code by the Code Tables in \cite{Codetable}.
\end{example}

\begin{remark}
  We remark that there are some other optimal or almost optimal locally recoverable codes obtained from this section expect for an infinite family of at least almost $k$-optimal locally recoverable codes in Theorem \ref{th-loc-q=2}. We list them in Table \ref{tab-optimal-loc}.
  \begin{table}[!h]
\begin{center}
\caption{Optimal or almost optimal locally recoverable codes derived in Section \ref{sec4}.}\label{tab-optimal-loc}
\begin{tabular}{llll}
\toprule
Conditions & Code & Parameters & Optimality \\
\midrule
$q=2, r=3$ & $\overline{\cC_D}(N=q^r+1)$ & $(28, 7, 12, 2; 3)$ & Almost $k$-optimal\\
$q=3, r=2$ & $\overline{\cC_D}(N=q^r+1)$ & $(21, 5, 12, 3; 2)$ & $k$-optimal\\
$q=4, r=2$ & $\overline{\cC_D}(N=q^r+1)$ & $(52, 5, 36, 4; 2)$ & Almost $k$-optimal\\
$q=3, m=3$ & $\overline{\cC_D}(N=2)$ & $(9,4,4,3;2)$ & $k$-optimal and almost $d$-optimal\\
\bottomrule
\end{tabular}
\end{center}
\end{table}

\end{remark}

\section{The second family of linear codes $\cC_f$ from weakly regular bent functions}\label{sec5}
In this subsection, let $q = p^e$ with odd prime $p$ and positive integer $e$. Let $f(x)$ denote a weakly regular bent function with $f(0)=0$. Let $\cC_f$ be the linear code defined in Equation (\ref{eq-Cfbar}). Firstly, we will prove that $\cC_f$ is a family of self-orthogonal codes and then determine its parameters and weight distribution. Besides, we will prove that $\cC_f$ is also a family of optimally extendable codes in most cases. Furthermore, two families of optimal locally recoverable codes are derived.

\begin{lemma}\label{lem-Ns2}
Let $s \in \gf_p^*$, $a \in \gf_p$ and $b \in \gf_{p^e}$. Let $f(x)\in \mathcal{RF}$ with $\varepsilon$ the sign of its Walsh transform. Let $N_s$ denote the number of solutions in $\gf_{p^e}$ of the equation $af(x) + \tr_{p^e/p}\left(bx\right)=s$.
If $e$ is even, then
\begin{eqnarray*}
N_s=
\begin{cases}
0, &\text{if $a=0,b=0$,}\\
p^{e-1},     &\text{if $a=0, b \in \gf_{p^e}^*$ ,}\\
p^{e-1}+\frac{\varepsilon(p-1)\sqrt{p^*}^e}{p},     &\text{if $a \in \gf_p^*, f^*(-\frac{b}{a})=a^{-1}s$,}\\
p^{e-1}-\frac{\varepsilon \sqrt{p^*}^e}{p},     &\text{if $a \in \gf_p^*, f^*(-\frac{b}{a})\neq a^{-1}s$,}
\end{cases}
\end{eqnarray*}
where $p^* = (-1)^{\frac{p-1}{2}}p$.
If $e$ is odd, then
\begin{eqnarray*}
N_s=
\begin{cases}
0, &\text{if $a=0,b=0$,}\\
p^{e-1},     &\text{if $a=0, b \in \gf_{p^e}^*$, or $a \in \gf_p^*, f^*(-\frac{b}{a})=a^{-1}s$, }\\
p^{e-1}+\frac{\varepsilon \sqrt{p^*}^eG(\eta_0, \lambda_1)}{p},     &\text{if $a \in \gf_p^*, f^*(-\frac{b}{a})\neq a^{-1}s$ and $\eta_0(f^*(-\frac{b}{a})- a^{-1}s)=1$,}\\
p^{e-1}-\frac{\varepsilon \sqrt{p^*}^eG(\eta_0, \lambda_1)}{p},     &\text{if $a \in \gf_p^*, f^*(-\frac{b}{a})\neq a^{-1}s$ and $\eta_0(f^*(-\frac{b}{a})- a^{-1}s)=-1$,}
\end{cases}
\end{eqnarray*}
where  $G(\eta_0, \lambda_1) = \sqrt{-1}^{(\frac{p-1}{2})^2}\sqrt{p}$.
\end{lemma}
\begin{proof}
 By the orthogonal relation of additive characters, we have
\begin{eqnarray*}
N_s &=&| \{x\in\gf_{p^e}: af(x) + \tr_{p^e/p}\left(bx\right)=s\}|\\
&=&\frac{1}{p}\sum_{y \in \gf_p}\sum_{x \in \gf_{p^e}}\zeta_p^{yaf(x)+\tr_{p^e/p}(yba)-ys}\\
&=&\frac{1}{p}\sum_{y \in \gf_p^*}\lambda_1(-ys)\sum_{x \in \gf_{p^e}}\zeta_p^{yaf(x)+\tr_{p^e/p}(ybx)}+p^{e-1}.
\end{eqnarray*}

In order to calculate the value of $N_s$, we will consider the following cases.

{Case 1:} If $a=0$, then
\begin{eqnarray*}
N_s=\frac{1}{p}\sum_{y \in \gf_p^*}\lambda_1(-ys)\sum_{x \in \gf_{p^e}}\phi_1(ybx)+p^{e-1}
=\begin{cases}
        p^{e-1}, & \mbox{if $b \in \gf_{p^e}^*$}, \\
        0, & \mbox{if $b=0$}.
      \end{cases}
\end{eqnarray*}

{Case 2:} If $a\neq 0$, then
\begin{eqnarray*}
N_s&=&\frac{1}{p}\sum_{y \in \gf_p^*}\lambda_1(-ys)\sum_{x \in \gf_{p^e}}\zeta_p^{ya\left(f(x)-\tr_{p^e/p}(-\frac{b}{a}x)\right)}+p^{e-1}\\
&=&\frac{1}{p}\sum_{y \in \gf_p^*}\lambda_1(-ys)\sigma_{ya}\left(\text{W}_f\left(-\frac{b}{a}\right)\right)+p^{e-1}\\
&=&\frac{1}{p}\sum_{y \in \gf_p^*}\lambda_1(-ya^{-1}s)\sigma_{y}\left(\text{W}_f\left(-\frac{b}{a}\right)\right)+p^{e-1},
\end{eqnarray*}
where $\sigma_y$ is the automorphism of $\mathbb{Q}(\zeta_p)$ defined by $\sigma_y(\zeta_p)=\zeta_p^y$ and $f(0)=0$. Since $f(x) \in \mathcal{RF}$, then
\begin{eqnarray*}
\sigma_{y}\left(\text{W}_f\left(-\frac{b}{a}\right)\right)=\sigma_{y}\left(\varepsilon \sqrt{p^*}^e\zeta_p^{f^*(-\frac{b}{a})}\right)
= \varepsilon \sqrt{p^*}^e\zeta_p^{yf^*(-\frac{b}{a})}\eta_0^e(y)
\end{eqnarray*}
by Equations (\ref{eq-Wf}) and (\ref{eq-sigma}), where $p^* = (-1)^{\frac{p-1}{2}}p$. Then
\begin{eqnarray*}
N_s &=&\frac{\varepsilon \sqrt{p^*}^e}{p}\sum_{y \in \gf_p^*}\lambda_1(-ya^{-1}s)\zeta_p^{yf^*\left(-\frac{b}{a}\right)}\eta_0^e(y)+p^{e-1}\\
&=& \frac{\varepsilon \sqrt{p^*}^e}{p}\sum_{y \in \gf_p^*}\lambda_1\left(y\left(f^*\left(-\frac{b}{a}\right)-a^{-1}s\right)\right)\eta_0^e(y)+p^{e-1}.
\end{eqnarray*}

{Subcase 2.1:} When $e$ is an even integer, by the orthogonality of additive characters, we have
\begin{eqnarray*}
N_s=\frac{\varepsilon \sqrt{p^*}^e}{p}\sum_{y \in \gf_p^*}\lambda_1\left(y\left(f^*\left(-\frac{b}{a}\right)-a^{-1}s\right)\right)+p^{e-1}
=\begin{cases}
     \frac{\varepsilon (p-1) \sqrt{p^*}^e}{p}+p^{e-1}, & \mbox{if $f^*(-\frac{b}{a})=a^{-1}s$}, \\
     -\frac{\varepsilon \sqrt{p^*}^e}{p}+p^{e-1}, & \mbox{if $f^*(-\frac{b}{a}) \neq a^{-1}s$}.
   \end{cases}
\end{eqnarray*}

{Subcase 2.2:} When $e$ is an odd integer, by the orthogonality of additive characters, we have
\begin{eqnarray*}
N_s&=&\frac{\varepsilon \sqrt{p^*}^e}{p}\sum_{y \in \gf_p^*}\lambda_1\left(y\left(f^*\left(-\frac{b}{a}\right)-a^{-1}s\right)\right)\eta_0(y)+p^{e-1}\\
&=&\begin{cases}
     p^{e-1}, & \mbox{if $f^*(-\frac{b}{a})=a^{-1}s$}, \\
     \frac{\varepsilon \sqrt{p^*}^e}{p}G(\eta_0, \lambda_1)+p^{e-1}, & \mbox{if $\eta_0(f^*(-\frac{b}{a})- a^{-1}s)=1$},\\
     -\frac{\varepsilon \sqrt{p^*}^e}{p}G(\eta_0, \lambda_1)+p^{e-1}, & \mbox{if $\eta_0(f^*(-\frac{b}{a})- a^{-1}s)=-1$},
   \end{cases}
\end{eqnarray*}
where $G(\eta_0, \lambda_1) = \sqrt{-1}^{(\frac{p-1}{2})^2}\sqrt{p}$ by Lemma \ref{quadGuasssum1}.

Summarizing the above discussions, the desired conclusions follow.
\end{proof}

\subsection{The self-orthogonality of $\cC_f$}

\begin{theorem}\label{tem-wtCD2}
Let $q = p^e$, where $p$ is an odd prime, $e$ is a positive integer with $e \geq 2$. Let $f(x) \in \mathcal{RF}$ with $\varepsilon = \pm 1$ the sign of its Walsh transform. Then the code $\cC_{f}$ defined in Equation (\ref{eq-Cfbar}) is a $[q, e+2]$ $p$-ary linear code and its weight distribution is listed in Tables \ref{tab3} and \ref{tab4} for even $e$ and for odd $e$, respectively. In particular, $\cC_f$ is a family of self-orthogonal codes for $e \geq 3$. Besides, $\cC_{f}^{\perp}$ is an almost MDS code with parameters $[q, q-4, 4]$ if $e =2$ and $\varepsilon \eta_0(-1)=-1$, and $\cC_{f}^{\perp}$ is optimal according to the sphere-packing bound in this case. In other cases, $\cC_{f}^{\perp}$ has parameters $[q, q-e-2, 3]$ and is at least almost optimal according to the sphere-packing bound.
\end{theorem}
\begin{table}[h!]
\begin{center}
\caption{The weight distribution of $\cC_{f}$ in Theorem \ref{tem-wtCD2} for even $e$.}\label{tab3}
\begin{tabular}{@{}ll@{}}
\toprule
Weight & Frequency  \\
\midrule
$0$ & $1$ \\
$p^e-p^{e-1}-\frac{\varepsilon (p-1)\sqrt{p^*}^e}{p}$ & $(p-1)p^e$  \\
$p^e-p^{e-1}+\frac{\varepsilon \sqrt{p^*}^e}{p}$ & $(p-1)^2p^e$ \\
$p^e-p^{e-1}$ & $ p(p^e-1)$ \\
$p^e$ &   $p-1$\\
\bottomrule
\end{tabular}
\end{center}
\end{table}
\begin{table}[h!]
\begin{center}
\caption{The weight distribution of $\cC_{f}$ in Theorem \ref{tem-wtCD2} for odd $e$.}\label{tab4}
\begin{tabular}{@{}ll@{}}
\toprule
Weight & Frequency  \\
\midrule
$0$ & $1$ \\
$p^e-p^{e-1}-\frac{\varepsilon \sqrt{p^*}^e G(\eta_0, \lambda_1)}{p}$ &  $\frac{(p-1)^2p^e}{2}$ \\
$p^e-p^{e-1}+\frac{\varepsilon \sqrt{p^*}^e G(\eta_0, \lambda_1)}{p}$ &  $\frac{(p-1)^2p^e}{2}$\\
$p^e-p^{e-1}$ &  $2p^{e+1}-p^e-p$ \\
$p^e$ &  $p-1$ \\
\bottomrule
\end{tabular}
\end{center}
\end{table}

\begin{proof}
For $c=0$, by the orthogonal relation of additive characters, the Hamming weight
\begin{eqnarray*}
  \text{wt}(\bc_{(a,b,0)}) &=& p^e-|\{x \in \gf_q: af(x)+\tr_{q/p}(bx)=0\}| \\
   &=&  p^e-\frac{1}{p}\sum_{y \in \gf_p} \sum_{x \in \gf_q}\phi_1(ybx) \lambda_1(yaf(x))\\
   &=& p^e-p^{e-1}-\frac{1}{p}\sum_{y \in \gf_p^*}\sum_{x \in \gf_q} \lambda_1(yaf(x)+y\tr_{q/p}(bx))\\
   &=& \begin{cases}
         0 & \mbox{if $a=0, b=0$},  \\
         p^e-p^{e-1} & \mbox{if $a=0, b \neq 0$}, \\
         p^e-p^{e-1}-\frac{1}{p}\sum_{y \in \gf_p^*}\sum_{x \in \gf_q}\zeta_p^{ya(f(x)-\tr_{q/p}(-\frac{b}{a}x))} & \mbox{if $a \neq 0$}.
       \end{cases}
\end{eqnarray*}
Denote by $N_1=p^e-p^{e-1}-\frac{1}{p}\sum_{y \in \gf_p^*}\sum_{x \in \gf_q}\zeta_p^{ya(f(x)-\tr_{q/p}(-\frac{b}{a}x))}$. Similarly to the proof of Lemma \ref{lem-Ns2}, we derive that
\begin{eqnarray*}
  N_1 &=& \begin{cases}
            p^e-p^{e-1}-\frac{\varepsilon(p-1)\sqrt{p^*}^e}{p} & \mbox{if $f^*(-\frac{b}{a})=0$},\\
            p^e-p^{e-1}+\frac{\varepsilon \sqrt{p^*}^e}{p} & \mbox{if $f^*(-\frac{b}{a})\neq 0$},
          \end{cases}
\end{eqnarray*}
if $e$ is an even integer,
and
\begin{eqnarray*}
  N_1 &=& \begin{cases}
            p^e-p^{e-1} & \mbox{if $f^*(-\frac{b}{a})=0$},\\
            p^e-p^{e-1}-\frac{\varepsilon \sqrt{p^*}^eG(\eta_0, \lambda_1)}{p} & \mbox{if $f^*(-\frac{b}{a})\neq 0$ and $\eta_0(f^*(-\frac{b}{a}))=1$},\\
            p^e-p^{e-1}+\frac{\varepsilon \sqrt{p^*}^eG(\eta_0, \lambda_1)}{p} & \mbox{if $f^*(-\frac{b}{a})\neq 0$ and $\eta_0(f^*(-\frac{b}{a}))=-1$},
          \end{cases}
\end{eqnarray*}
if $e$ is an odd integer.
For $c \in \gf_p^*$, denote by $N_{(a,b,c)}=|\{x \in \gf_q:af(x)+\tr_{q/p}(bx)+c=0\}|.$ Then the Hamming weight of codeword $\bc_{(a,b,c)}$ is $\text{wt}(\bc_{(a,b,c)})=p^e-N_{(a,b,c)}$ for $a \in \gf_p, b \in \gf_{p^e}, c \in \gf_p^*$. Besides, we can verify that the value of $N_{(a,b,c)}$ is equal to the value of $N_s$ in Lemma \ref{lem-Ns2} for $s=-c$. Then by Lemma \ref{lem-Ns2}, if $e$ is an even integer, we have
\begin{eqnarray*}
\text{wt}(\bc_{(a,b,c)})&=&p^e-N_{(a,b,c)}\\
&=&\begin{cases}
p^e &\text{if $a=0,b=0$,}\\
p^e-p^{e-1}     &\text{if $a=0, b \in \gf_{p^e}^*$,}\\
p^e-p^{e-1}-\frac{\varepsilon(p-1)\sqrt{p^*}^e}{p}     &\text{if $a \in \gf_p^*, f^*(-\frac{b}{a})+a^{-1}c=0$,}\\
p^e-p^{e-1}+\frac{\varepsilon\sqrt{p^*}^e}{p}     &\text{if $a \in \gf_p^*, f^*(-\frac{b}{a})+ a^{-1}c \neq 0$.}
\end{cases}
\end{eqnarray*}
If $e$ is an odd integer, we have
\begin{eqnarray*}
\text{wt}(\bc_{(a,b,c)})&=&p^e-N_{(a,b,c)}\\
&=&\left\{\begin{array}{lll}
p^e \qquad\qquad\quad \mbox{if $a=0,b=0$,}\\
p^e-p^{e-1}    \quad \quad \mbox{if $a=0, b \in \gf_{p^e}^*$, or $a \in \gf_p^*, f^*(-\frac{b}{a})+a^{-1}c=0$,} \\
p^e-p^{e-1}-\frac{\varepsilon \sqrt{p^*}^eG(\eta_0, \lambda_1)}{p} \\
\qquad\qquad\qquad \thinspace \mbox{if $a \in \gf_p^*, f^*(-\frac{b}{a})+ a^{-1}c \neq 0$ and $\eta_0(f^*(-\frac{b}{a})+ a^{-1}c)=1$,}\\
p^e-p^{e-1}+\frac{\varepsilon \sqrt{p^*}^eG(\eta_0, \lambda_1)}{p}  \\
\qquad\qquad\qquad \thinspace \mbox{if $a \in \gf_p^*, f^*(-\frac{b}{a})+ a^{-1}c \neq 0$ and $\eta_0(f^*(-\frac{b}{a})+ a^{-1}c)=-1$.}
\end{array} \right.
\end{eqnarray*}

Summarizing the above discussions, if $e$ is an even integer,
\begin{eqnarray*}
\text{wt}(\bc_{(a,b,c)})&=&\begin{cases}
0 &\text{if $a=0,b=0,c=0$,}\\
p^e &\text{if $a=0, b =0, c \in \gf_p^*$,}\\
p^e-p^{e-1}     &\text{if $a=0, b \in \gf_{p^e}^*, c \in \gf_p$, }\\
p^e-p^{e-1}-\frac{\varepsilon(p-1)\sqrt{p^*}^e}{p}    &\text{if $a \in \gf_p^*, f^*(-\frac{b}{a})+a^{-1}c=0$,}\\
p^e-p^{e-1}+\frac{\varepsilon \sqrt{p^*}^e}{p}     &\text{if $a \in \gf_p^*, f^*(-\frac{b}{a})+ a^{-1}c \neq 0$.}
\end{cases}
\end{eqnarray*}
If $e$ is an odd integer,
\begin{eqnarray*}
\text{wt}(\bc_{(a,b,c)})&=&\left\{\begin{array}{lll}
0 \qquad\qquad \mbox{if $a=0,b=0,c=0$,}\\
p^e \quad\qquad \enspace \mbox{if $a=0, b=0, c \in \gf_p^*$,}\\
p^e-p^{e-1}   \quad  \mbox{if $a=0, b \in \gf_{p^e}^*, c \in \gf_p$, or $a \in \gf_p^*, f^*(-\frac{b}{a})+a^{-1}c=0$},\\
p^e-p^{e-1}-\frac{\varepsilon \sqrt{p^*}^eG(\eta_0, \lambda_1)}{p}  \\  \qquad\qquad\quad \mbox{if $a \in \gf_p^*, f^*(-\frac{b}{a})+a^{-1}c \neq 0$ and $\eta_0(f^*(-\frac{b}{a})+a^{-1}c)=1$},\\
p^e-p^{e-1}+\frac{\varepsilon \sqrt{p^*}^eG(\eta_0, \lambda_1)}{p}  \\   \qquad\qquad\quad \mbox{if $a \in \gf_p^*, f^*(-\frac{b}{a})+ a^{-1}c \neq 0$ and $\eta_0(f^*(-\frac{b}{a})+a^{-1}c)=-1$.}
\end{array} \right.
\end{eqnarray*}

When $e$ is even, we denote by $w_1=p^e$, $w_2=p^e-p^{e-1}$, $w_3=p^e-p^{e-1}-\frac{\varepsilon(p-1)\sqrt{p^*}^e}{p}$, $w_4=p^e-p^{e-1}+\frac{\varepsilon\sqrt{p^*}^e}{p}$. Next we will determine the frequency $A_{w_i}, 1 \leq i \leq 4$. It is easy to deduce that $A_{w_1}=p-1$ and $A_{w_2}=p(p^{e}-1)$. Note that $$A_{w_3}=\left|\left\{(a,b,c) \in \gf_p^* \times \gf_{p^e} \times \gf_p: f^*\left(-\frac{b}{a}\right)+a^{-1}c=0\right\}\right|.$$ For fixed $a$ and $b$, then the equation $f^*\left(-\frac{b}{a}\right)+a^{-1}c=0$ with variable $c$ has the unique solution. Hence, the number of $(a,b,c)$ is $(p-1)p^e$ and $A_{w_3}=(p-1)p^e$. Since $$A_{w_4}=\left|\left\{(a,b,c) \in \gf_p^* \times \gf_{p^e} \times \gf_p: f^*\left(-\frac{b}{a}\right)+a^{-1}c\neq 0\right\}\right|,$$ we deduce $A_{w_4}=p^{e+1}(p-1)-A_{W_3}=(p-1)^2p^e$.

When $e$ is odd, we denote by $w_1=p^e$, $w_2=p^e-p^{e-1}$, $w_3=p^e-p^{e-1}-\frac{\varepsilon \sqrt{p^*}^eG(\eta_0, \lambda_1)}{p}$, $w_4=p^e-p^{e-1}+\frac{\varepsilon \sqrt{p^*}^eG(\eta_0, \lambda_1)}{p}$. Next we will determine the frequency $A_{w_i}, 1 \leq i \leq 4$. It is easy to deduce that $A_{w_1}=p-1$ and $A_{w_2}=p(p^{e}-1)+p^e(p-1)=2p^{e+1}-p^e-p$. Note that $$A_{w_3}=\left|\left\{(a,b,c) \in \gf_p^* \times \gf_{p^e} \times \gf_p: f^*\left(-\frac{b}{a}\right)+a^{-1}c\neq 0, \eta_0\left(f^*\left(-\frac{b}{a}\right)+a^{-1}c\right)=1\right\}\right|.$$ For fixed $a\in \gf_p^*$, $b\in \gf_{p^e}$ and $\delta \in \langle \gamma^2 \rangle$, the solution of the equation $f^*(-\frac{b}{a})+a^{-1}c=\delta$ with variable $c$ is unique. Then the number of $(a,b,c)$ is $\frac{(p-1)^2p^e}{2}$ and $A_{w_3}=\frac{(p-1)^2p^e}{2}$. Similarly, $A_{w_4}=\frac{(p-1)^2p^e}{2}$.

Then the weight distribution of $\cC_f$ follows. We also deduce that $\mathbf{1} \in \cC_f$ and $\cC_f$ is $p$-divisible for $e \geq 3$ by the above discussions. Then by Lemma \ref{lem-self-orthogonal}, $\cC_{f}$ is a self-orthogonal code.

Finally, we will determine the parameters of $\cC_{f}^{\perp}$. If $e$ is even, by the first four Pless power moments in \cite[Page 260]{H}, we have
\begin{eqnarray*}
A_1^{\perp}=A_2^{\perp}=0,\ A_3^{\perp}=\frac{p^{e-1}(p-1)\left(p(p^e-2p^{e-1}-p+2)+ \varepsilon \eta_0^{\frac{e}{2}}(-1)p^{\frac{e}{2}}(p^2-3p+2)\right)}{6}.
\end{eqnarray*}
When $e =2$ and $\varepsilon \eta_0(-1)=-1$, it is easy to deduce that $A_3^{\perp} = 0$ and $$A_4^{\perp}=\frac{p^2(p-1)^2(p^4-p^3-5p^2+3p+6)}{24} > 0$$ by the fifth Pless power moments in \cite[Page 260]{H}. In this case, the parameters of $\cC_{f}^{\perp}$ are $[p^2, p^2-4, 4]$. In other cases, the parameters of $\cC_{f}^{\perp}$ are $[p^e, p^e-e-2, 3]$. If $e$ is odd, by the first four Pless power moments in \cite[Page 260]{H}, we have
\begin{eqnarray*}
A_1^{\perp}=A_2^{\perp}=0,\ A_3^{\perp}=\frac{p^{e}(p-1)(p^e-2p^{e-1}-p+2)}{6}.
\end{eqnarray*}
Note that $A_3^{\perp} > 0$ as $e \geq 3$ is an odd integer. Then the parameters of $\cC_{f}^{\perp}$ is $[p^e, p^e-e-2, 3]$.
The proof is completed by the sphere packing bound in Lemma \ref{sphere}.
\end{proof}

The following corollary shows that $\cC_f$ holds 2-designs in some special cases.
\begin{corollary}\label{col-design2}
  Let $q=p^2$ and $f(x) \in \cRF$ with $\varepsilon = \pm 1$ the sign of its Walsh transform, where $p$ is an odd prime. Let $\eta_0$ be the quadratic multiplicative character of $\gf_p$ and $\varepsilon \eta_0(-1)=-1$. Then the linear code $\cC_{f}$ has parameters $[p^2, 4, p^2-p-1]$ and weight enumerator
  \begin{eqnarray*}
  A(z)=1+p^2(p-1)^2z^{p^2-p-1}+p(p^2-1)z^{p^2-p}+p^2(p-1)z^{p^2-1}+(p-1)z^{p^2}.
  \end{eqnarray*}
  The dual code $\cC_{f}^{\perp}$ has parameters $[p^2, p^2-4, 4]$. Besides, the minimum weight codewords in $\cC_{f}$ hold a $2$-$(p^2, p^2-p-1, (p-2)(p^2-p-1))$ design and its complementary design is a  $2$-$(p^2, p+1, p)$ design. The minimum weight codewords in $\cC_{f}^{\perp}$ hold a $2$-$(p^2, 4, \frac{(p^2-3)(p-2)}{2})$ design.
\end{corollary}
\begin{proof}
  The desired conclusions follow from the proof of Theorem \ref{tem-wtCD2} and the Assmus-Mattson Theorem \cite{H}.
\end{proof}
\subsection{Optimally extendable codes}
In this subsection, we will prove that $\cC_{f}$ is an almost optimally or optimally extendable code if we select a suitable generator matrix.

Let $\gf_{q}^*=\langle\beta\rangle$ for $q=p^e$. Then $\{1,\beta, \beta^{2},\ldots, \beta^{e-1}\}$ is a $\gf_p$-basis of $\gf_{q}$.
Let $\cC_f$ be the code defined in Theorem \ref{tem-wtCD2}. Let $d_1,d_2,\cdots,d_{q-1}, d_q$ denote all elements of $\gf_q$, where $d_q=0$.
By definition, the generator matrix $G$ of $\cC_f$ is given by
\begin{eqnarray*}
G:=\left[
\begin{array}{cccc}
1&1&\cdots&1 \\
 f(d_1) & f(d_2) & \cdots &f(d_q)\\
 \tr_{q/p}(\beta^{0}d_{1})& \tr_{q/p}(\beta^{0}d_{2})& \cdots &\tr_{q/p}(\beta^{0}d_{q}) \\
\tr_{q/p}(\beta^{1}d_{1})& \tr_{q/p}(\beta^{1}d_{2})& \cdots &\tr_{q/p}(\beta^{1}d_{q}) \\
\vdots &\vdots &\ddots &\vdots \\
\tr_{q/p}(\beta^{e-1}d_{1})& \tr_{q/p}(\beta^{e-1}d_{2})& \cdots &\tr_{q/p}(\beta^{e-1}d_{q})
\end{array}\right].
\end{eqnarray*}
However, we find that the dual of the linear code $\widetilde{\cC_f'}$ generated by the matrix $G':=[I_{e+2,e+2}:G]$ has minimum distance 2.
By an elementary row transformation on $G$, we obtain another generator matrix $G_2$ of $\cC_f$  given by
\begin{eqnarray}\label{matrix2}
G_2:=\left[
\begin{array}{cccc}
1&1&\cdots&1 \\
 f(d_1)+1 & f(d_2)+1 & \cdots &f(d_q)+1\\
 \tr_{q/p}(\beta^{0}d_{1})& \tr_{q/p}(\beta^{0}d_{2})& \cdots &\tr_{q/p}(\beta^{0}d_{q}) \\
\tr_{q/p}(\beta^{1}d_{1})& \tr_{q/p}(\beta^{1}d_{2})& \cdots &\tr_{q/p}(\beta^{1}d_{q}) \\
\vdots &\vdots &\ddots &\vdots \\
\tr_{q/p}(\beta^{e-1}d_{1})& \tr_{q/p}(\beta^{e-1}d_{2})& \cdots &\tr_{q/p}(\beta^{e-1}d_{q})
\end{array}\right].
\end{eqnarray}

Let $\cC_f'$ be the linear code with generator matrix $[I_{e+2, e+2}: G_2]$. We first derive the parameters of $\cC_{f}'$ and $\cC_{f}'^{\perp}$ for even $e$.

\begin{theorem}\label{th-LCD2}
  Let $q=p^e$, where $p$ is an odd prime and $e \geq 2$ is an even integer. Let $f(x) \in \mathcal{RF}$ with $\varepsilon$ the sign of the Walsh transform of $f(x)$. Let $G_2$ be defined above and $G_2':=[I_{e+2, e+2}: G_2]$. If $\varepsilon (\eta_0(-1))^{\frac{e}{2}}=1$, then the linear code $\cC_f'$ generated by $G_2'$ has parameters
  $$\left[p^e+e+2, e+2, p^e-p^{e-1}-(p-1)p^{\frac{e}{2}-1}+2\right].$$
  If $\varepsilon (\eta_0(-1))^{\frac{e}{2}}=-1$, then the linear code $\cC_f'$ generated by $G_2'$ has parameters
  $$\left[p^e+e+2, e+2, p^e-p^{e-1}-p^{\frac{e}{2}-1}+1\right].$$
Besides, $\cC_f'^{\perp}$ has parameters $[p^e+e+2, p^e, 3]$ and is at least almost optimal according to the sphere-packing bound.
\end{theorem}

\begin{proof}
Obviously, the length of $\cC_f'$ is $q+e+2$ and the dimension of $\cC_f'$ is $e+2$. Similarly to the proof of Theorem \ref{th-LCD1}, the minimum distance $d$ of $\cC_f'$ satisfies that $d \geq d(\cC_{f})+1$, where $d(\cC_{f})$ denotes  the minimum distance of $\cC_{f}$ and was given in Theorem \ref{tem-wtCD2}. For a codeword $\ba'$ in $\cC_f'$, we have
\begin{eqnarray*}
\ba'&=&(c,a,a_0,a_1 \cdots, a_{e-1})G_1'\\&=&(c,a,a_0, \cdots, a_{e-1}, af(d_1)+\sum_{i=0}^{e-1}a_i\tr_{q/p}(\beta^{i}d_{1})+c+a, \cdots, af(d_q)+\sum_{i=0}^{e-1}a_i\tr_{q/p}(\beta^{i}d_{q})+c+a),
\end{eqnarray*}
where $c, a, a_i \in \gf_p$ for all $0 \leq i \leq m-1$. Then $\ba'=(c,a,a_0,a_1 \cdots, a_{e-1}, \ba)$, where $$\ba=(c,a,a_0,a_1 \cdots, a_{e-1})G_1=(af(x)+\tr_{p^e/p}(bx) + c+a)_{x \in \gf_q}$$ is a codeword in $\cC_f$, where $b:=\sum_{i=0}^{e-1}a_i\beta^i$. Now we consider the following two cases:

{Case 1:} Let $\varepsilon (\eta_0(-1))^{\frac{e}{2}}=1$.
By Theorem \ref{tem-wtCD2},  we deduce $\text{wt}(\ba)=d(\cC_f)=p^e-p^{e-1}-(p-1)p^{\frac{e}{2}-1}$ if and only if $a \in \gf_p^*$ and $f^*(-\frac{b}{a})+a^{-1}(c+a)=0$. Note that $f^*(-\frac{b}{a})+a^{-1}(c+a)=0$ implies $(b,c)\neq (0,0)$. Besides, $b=0$ if and only if $(a_0,a_1,\cdots,a_{e-1})=(0,0,\cdots,0)$.
Then $d \geq d(\cC_{f})+2$.
In particular, if $a \in \gf_p^*$, $b=0$ and $c=-a$, then $f^*(-\frac{b}{a})+a^{-1}(c+a)=0$ and $\text{wt}(\ba')=\text{wt}(\ba)+2=d(\cC_{f})+2$.
Hence $d=d(\cC_{f})+2$.

{Case 2:} Let $\varepsilon (\eta_0(-1))^{\frac{e}{2}}=-1$. By Theorem \ref{tem-wtCD2},  we deduce $\text{wt}(\ba)=d(\cC_f)=p^e-p^{e-1}-p^{\frac{e}{2}-1}$ if and only if $a \in \gf_p^*$ and $f^*(-\frac{b}{a})+a^{-1}(c+a)\neq0$.
If $a\neq 0$, $b=0$ and $c=0$, then $f^*(-\frac{b}{a})+a^{-1}(c+a)=1$ and $\text{wt}(\ba')=\text{wt}(\ba)+1=d(\cC_{f})+1$.
Hence $d=d(\cC_{f})+1=p^e-p^{e-1}-p^{\frac{e}{2}-1}+1$.

 It is clear that $\cC_f'^{\perp}$ has length $q+e+2$ and dimension $q$. In the following we prove that the minimum distance $d^\perp$ of $\cC_f'^{\perp}$ is $3$.
Note that $G_2'$ is a check matrix of $\cC_f'^{\perp}$ and any column vector in $G_2'$ is nonzero. Hence, $d^{\perp} \geq 2$.
Note that any two columns of the first $e+2$ columns of matrix $G_2'$ are $\gf_p$-linearly independent and any two columns of the last $q$ columns of matrix $G_2'$ are also $\gf_p$-linearly independent. Select any column $\bg_1$ in the first $e+2$ columns of matrix $G_2'$ and any column $\bg_2$ in the last $q$ columns of matrix $G_2'$. If $\bg_1 \neq (1,0,\cdots,0)$, it is easy to deduce that $\bg_1$ and $\bg_2$ are $\gf_p$-linearly independent. If $\bg_1 = (1,0,\cdots,0)$, then $\bg_1$ and $\bg_2$ are $\gf_p$-linearly dependent if and only if
\begin{eqnarray}\label{system}
\left\{\begin{array}{c}
  f(d_i)+1=0, \\
  \tr_{q/p}(\beta^{0}d_{i})=0, \\
  \tr_{q/p}(\beta^{1}d_{i})=0, \\
  \vdots \\
  \tr_{q/p}(\beta^{e-1}d_{i})=0,
\end{array}\right.
\end{eqnarray}
where $d_i \in \gf_q$. This system of equations implies $d_i\in \gf_q^*$.
For any $x=\sum_{l=0}^{e-1}k_l\beta^l \in \gf_{q}, k_l \in \gf_p$, the System (\ref{system}) implies
$$
\tr_{q/p}(xd_i)=\sum_{l=0}^{e-1}k_l \tr_{q/p}(\beta^ld_i)=0.
$$
This contradicts with the fact that $|\ker(\tr_{q/p})|=p^{e-1}$.
Hence, $\bg_1$ and $\bg_2$ are $\gf_p$-linearly independent and $d^{\perp} \geq 3$. Besides, the first, second and last columns of $G_2'$ are $\gf_p$-linearly dependent. Then $d^{\perp}=3$.
\end{proof}

We then derive the parameters of $\cC_{f}'$ and $\cC_{f}'^{\perp}$ for odd $e$.

\begin{theorem}\label{th-LCD3}
  Let $q=p^e$, where $p$ is an odd prime and $e \geq 3$ is an odd integer. Let $f(x) \in \mathcal{RF}$ with $\varepsilon$ the sign of its Walsh transform. Let $G_2$ be defined above and $G_2':=[I_{e+2, e+2}: G_2]$. If $\varepsilon (\eta_0(-1))^{\frac{e+1}{2}}=1$, then the linear code $\cC_{f}'$ generated by matrix $G_2'$ has parameters $$
    \left[p^e+e+2, e+2, p^e-p^{e-1}-p^{\frac{e-1}{2}}+1\right].$$
  If $\varepsilon (\eta_0(-1))^{\frac{e+1}{2}}=-1$, then the linear code $\cC_{f}'$ generated by matrix $G_2'$ has parameters
  $$
    \left[p^e+e+2, e+2, p^e-p^{e-1}-p^{\frac{e-1}{2}}+2\right].$$
Besides, $\cC_{f}'^{\perp}$ has parameters $[p^e+e+2, p^e, 3]$ and is at least almost optimal according to the sphere-packing bound.
\end{theorem}
\begin{proof}
  Similarly to the proof of Theorem \ref{th-LCD2}, the desired conclusions can be derived.
\end{proof}

\begin{theorem}\label{th-extend}
  Let $\cC_{f}$ be the linear code in Theorem \ref{tem-wtCD2} and $\cC_f'$ be the linear code with generator matrix $[I_{e+2, e+2}: G_2]$. If $e =2$ and $\varepsilon \eta_0(-1)=-1$, then $\cC_f$ is an almost optimally extendable self-orthogonal code. In all other cases, $\cC_f$ is an optimally extendable self-orthogonal code.
\end{theorem}

\begin{proof}
By Theorems \ref{tem-wtCD2}, \ref{th-LCD2} and \ref{th-LCD3}, we have $d(\cC_{f}^{\perp})-d(\cC_{f}'^{\perp}) = 1$ if $e =2$ and $\varepsilon \eta_0(-1)=-1$. In all other cases, $d(\cC_{f}^{\perp})-d(\cC_{f}'^{\perp}) = 0$. Then the desired conclusions follow.
\end{proof}

By Magma, we give two examples in the following to verify the results in Theorem \ref{th-extend}.
\begin{example}\label{exa-3}
Let $p=3, e=4$ and $f(x)=\tr_{3^e/3}(x^2)$. The sign of the Walsh transform of $f(x)$ is $\varepsilon = -1$. By Magma, then the ternary linear codes $\cC_{f}$, $\cC_{f}'$, $\cC_{f}^{\perp}$, $\cC_{f}'^{\perp}$ have parameters $$[81, 6, 51], [87, 6, 52], [81, 75, 3], [87, 81, 3],$$ respectively.
Besides, $\cC_{f}$ is an optimally extendable self-orthogonal code, and all of $\cC_{f}$, $\cC_{f}^{\perp}$ and $\cC_{f}'^{\perp}$ are optimal by the Code Tables in \cite{Codetable}.
\end{example}

\begin{example}\label{exa-4}
Let $p=3, e=3$ and $f(x)=\tr_{3^e/3}(wx^2)$, where $w$ is a primitive element of $\gf_{3^3}$. The sign of the Walsh transform of $f(x)$ is $\varepsilon = -1$. By Magma, then the ternary linear codes $\cC_{f}$, $\cC_{f}'$, $\cC_{f}^{\perp}$, $\cC_{f}'^{\perp}$ have parameters $$[27, 5, 15], [32, 5, 17], [27, 22, 3], [32, 27, 3],$$ respectively.
Besides, $\cC_{f}$ is an optimally extendable self-orthogonal code, and both of $\cC_{f}^{\perp}$ and $\cC_{f}'^{\perp}$ are optimal by the Code Tables in \cite{Codetable}.
\end{example}

\subsection{Optimal locally recoverable codes}
In this subsection, we will first determine the locality of $\cC_f$ and $\cC_f^{\perp}$ under certain conditions. Then we will prove that the locally recoverable codes $\cC_f$ and $\cC_f^{\perp}$ are optimal under certain conditions.

\begin{theorem}\label{loc-cf}
 Let $q=p^2$ and $f(x) \in \cRF$ with $\varepsilon = \pm 1$ the sign of its Walsh transform, where $p$ is an odd prime. Let $\eta_0$ be the quadratic multiplicative character of $\gf_p$ and $\varepsilon \eta_0(-1)=-1$. Then $\cC_{f}$ is a $(p^2,4,p^2-p-1,p;3)$-LRC and $\cC_f^{\perp}$ is a $(p^2,p^2-4,4,p;p^2-p-2)$-LRC. Besides, $\cC_f$ is $k$-optimal, and $\cC_f^{\perp}$ is both $k$-optimal and $d$-optimal.
\end{theorem}
\begin{proof}
By Corollary \ref{col-design2} and Lemma \ref{lem-locality}, we derive that the locality of $\cC_f$ is $3$ and the locality of $\cC_f^{\perp}$ is $p^2-p-2$. Now we prove $\cC_f$ is a $k$-optimal LRC.
By Lemma \ref{lem-CMbound}, putting $t=1$ and the parameters of the $(p^2,4,p^2-p-1,p;3)-$LRC into the right-hand side of the Cadambe-Mazumdar bound in (\ref{eqn-CMbound}), we have
\begin{eqnarray*}
  k &\leq& r+k_{opt}^{(p)}(n-(r+1),d) \\
   &=& 3+k_{opt}^{(p)}(p^2-4,p^2-p-1).
\end{eqnarray*}
Thanks to the Plotkin bound, we have
\begin{eqnarray*}
  k_{opt}^{(p)}(p^2-4,p^2-p-1) &\leq& \left\lfloor \log_p \left\lfloor \frac{p^2-p-1}{p^2-p-1-(1-\frac{1}{p})(p^2-4)} \right\rfloor \right\rfloor\\
  &=& \left\lfloor \log_p \left\lfloor \frac{p^2-p-1}{3-\frac{4}{p}} \right\rfloor \right\rfloor\\
  &=& \left\lfloor \log_p \left\lfloor \frac{p(p-1-\frac{1}{p})}{3-\frac{4}{p}} \right\rfloor \right\rfloor\\
  &=& 1,
\end{eqnarray*}
where the last equality holds as $p \leq \frac{p(p-1-\frac{1}{p})}{3-\frac{4}{p}} < p^2$. Thus, $k \leq 4$ and $\cC_f$ is $k$-optimal. Then we prove that $\cC_f^{\perp}$ is $k$-optimal and $d$-optimal. By Lemma \ref{lem-Slbound}, putting the parameters of the $(p^2,p^2-4,4,p;p^2-p-2)$-LRC into the right-hand side of the Singleton-like bound in (\ref{eqn-Slbound}), we have
$$
n-k- \left \lceil \frac{k}{r} \right \rceil +2
=p^2-(p^2-4)- \left \lceil \frac{p^2-4}{p^2-p-2} \right \rceil +2=4.
$$
Hence $\cC_f^{\perp}$ is a $d$-optimal LRC.
By Lemma \ref{lem-CMbound}, putting $t=1$ and the parameters of the $(p^2,p^2-4,4,p;p^2-p-2)$-LRC into the right-hand side of the Cadambe-Mazumdar bound in (\ref{eqn-CMbound}), we have
\begin{eqnarray*}
  k  &\leq& r+k_{opt}^{(p)}(n-(r+1),d)\\
  &=& p^2-p-2+k_{opt}^{(p)}(p+1,4)\\
  &\leq&p^2-p-2+p-2\\
  &=&p^2-4,
\end{eqnarray*}
where the third inequality holds as $k_{opt}^{(p)}(p+1,4) \leq p-2$ by the classical Singleton bound. Thus, $\cC_f^{\perp}$ is a $k$-optimal LRC. Then we have proved that $\cC_f^{\perp}$ is both $d$-optimal and $k$-optimal.
\end{proof}

In other cases, we can not determine the locality of $\cC_{f}$. However, by Magma, we find that the minimum weight codewords in $\cC_f$ (or $\cC_f^{\perp}$) hold a $1$-design in these cases.
\begin{example}
Let $p=3, e=3$ and $f(x)=\tr_{p^e/p}(x^2)$. By Magma, then $(\mathcal{P}(\cC_f), \mathcal{B}_{15}(\cC_f))$ is a $1$-$(27, 15, 15)$ design and $(\mathcal{P}(\cC_f^{\perp}), \mathcal{B}_{3}(\cC_f^{\perp}))$ is a $1$-$(27, 3, 4)$ design.
\end{example}

\begin{example}
Let $p=3, e=5$, $\gf_{p^e}^*=\langle\omega\rangle$ and $f(x)=\tr_{p^e/p}(\omega x^2)$. By Magma, then $(\mathcal{P}(\cC_f), \mathcal{B}_{153}(\cC_f))$ is a $1$-$(243, 153, 153)$ design and $(\mathcal{P}(\cC_f^{\perp}), \mathcal{B}_{3}(\cC_f^{\perp}))$ is a $1$-$(243, 3, 40)$ design.
\end{example}

\begin{example}
Let $p=5, e=2$, $\gf_{p^e}^*=\langle\omega\rangle$ and $f(x)=\tr_{p^e/p}(\omega x^2)$. By Magma, then $(\mathcal{P}(\cC_f), \mathcal{B}_{16}(\cC_f))$ is a $1$-$(25, 16, 16)$ design and $(\mathcal{P}(\cC_f^{\perp}), \mathcal{B}_{3}(\cC_f^{\perp}))$ is a $1$-$(25, 3, 12)$ design.
\end{example}

\begin{example}
Let $p=5, e=3$ and $f(x)=\tr_{p^e/p}(x^2)$. By Magma, then $(\mathcal{P}(\cC_f), \mathcal{B}_{95}(\cC_f))$ is a $1$-$(125, 95, 190)$ design and $(\mathcal{P}(\cC_f^{\perp}), \mathcal{B}_{3}(\cC_f^{\perp}))$ is a $1$-$(125, 3, 36)$ design.
\end{example}

\begin{example}
Let $p=7, e=2$ and $f(x)=\tr_{p^e/p}(x^2)$. By Magma, then $(\mathcal{P}(\cC_f), \mathcal{B}_{36}(\cC_f))$ is a $1$-$(49, 36, 36)$ design and $(\mathcal{P}(\cC_f^{\perp}), \mathcal{B}_{3}(\cC_f^{\perp}))$ is a $1$-$(49, 3, 30)$ design.
\end{example}

\begin{example}
Let $p=3, e=4$ and $f(x)=\tr_{p^e/p}(x^{p^{3k}+p^{2k}-p^k+1}+x^2)$. By Magma, then $(\mathcal{P}(\cC_f), \mathcal{B}_{51}(\cC_f))$ is a $1$-$(81, 51, 102)$ design and $(\mathcal{P}(\cC_f^{\perp}), \mathcal{B}_{3}(\cC_f^{\perp}))$ is a $1$-$(81, 3, 10)$ design.
\end{example}

Then we have the following conjecture by Lemma \ref{lem-locality}.
\begin{conj}\label{conj}
  Let $\cC_f$ be the linear code in Theorem \ref{tem-wtCD2} such that $(e,\varepsilon \eta_0(-1))\neq (2,-1)$. Then the minimum weight codewords in $\cC_f$ (or $\cC_f^{\perp}$) hold a $1$-design.
\end{conj}

  \begin{corollary}\label{coro}
  If Conjecture \ref{conj} holds, then we have the following:
  \begin{itemize}
    \item if $e$ is odd, then $\cC_f$ is a $(p^e, e+2, p^e-p^{e-1}-p^{\frac{e-1}{2}}, p; 2)$-LRC and $\cC_f^{\perp}$ is a $k$-optimal or almost $k$-optimal $(p^e, p^e-e-2, 3, p; p^e-p^{e-1}-p^{\frac{e-1}{2}}-1)$-LRC;
    \item if $e$ is even and $\varepsilon(\eta_0(-1))^{\frac{e}{2}}=1$, then $\cC_f$ is a $(p^e, e+2, p^e-p^{e-1}-(p-1)p^{\frac{e-2}{2}}, p; 2)$-LRC, $\cC_f^{\perp}$ is an almost $d$-optimal and $k$-optimal $(p^2, p^2-4, 3, p; p^2-2p)$-LRC for $e = 2$ and $\cC_f^{\perp}$ is a $k$-optimal or almost $k$-optimal $(p^e, p^e-e-2, 3, p; p^e-p^{e-1}-(p-1)p^{\frac{e-2}{2}}-1)$-LRC for $e \geq 4$;
    \item if $e \geq 4$ is even and $\varepsilon(\eta_0(-1))^{\frac{e}{2}}=-1$, then $\cC_f$ is a $(p^e, e+2, p^e-p^{e-1}-p^{\frac{e-2}{2}}, p; 2)$-LRC and $\cC_f^{\perp}$ is a $k$-optimal or almost $k$-optimal $(p^e, p^e-e-2, 3, p; p^e-p^{e-1}-p^{\frac{e-2}{2}}-1)$-LRC.
  \end{itemize}
\end{corollary}

\begin{proof}
If the minimum weight codewords in $\cC_f$ (or $\cC_f^{\perp}$) hold a $1$-design, then
  by Lemma \ref{lem-locality}, we derive that the locality of $\cC_f$ is $2$, the locality of $\cC_f^{\perp}$ is $p^e-p^{e-1}-p^{\frac{e-1}{2}}-1$ for odd $e$, the locality of $\cC_f^{\perp}$ is $p^e-p^{e-1}-(p-1)p^{\frac{e-2}{2}}-1$ for even $e$, and $\varepsilon(\eta_0(-1))^{\frac{e}{2}}=1$, the locality of $\cC_f^{\perp}$ is $p^e-p^{e-1}-p^{\frac{e-2}{2}}-1$ for even $e$, and $\varepsilon(\eta_0(-1))^{\frac{e}{2}}=-1$.

  Now we prove $\cC_f^{\perp}$ is a $k$-optimal or almost $k$-optimal LRC for odd $e$.
By Lemma \ref{lem-CMbound}, putting the parameters of the $(p^e, p^e-e-2, 3, p; p^e-p^{e-1}-p^{\frac{e-1}{2}}-1)$-LRC into the right-hand side of the Cadambe-Mazumdar bound in (\ref{eqn-CMbound}), we have
\begin{eqnarray*}
  k &\leq& \mathop{\min}_{t \in \mathbb{Z}^{+}} [rt+k_{opt}^{(q)}(n-t(r+1),d)]\\
  &\leq& \mathop{\min}_{t=1}[rt+k_{opt}^{(q)}(n-t(r+1),d)]\\
  &=& r+k_{opt}^{(q)}(n-(r+1),d) \\
   &=& p^e-p^{e-1}-p^{\frac{e-1}{2}}-1+k_{opt}^{(p)}(p^{e-1}+p^{\frac{e-1}{2}},3).
\end{eqnarray*}
Thanks to the sphere-packing bound, we have
\begin{eqnarray*}
  k_{opt}^{(p)}(p^{e-1}+p^{\frac{e-1}{2}},3) &\leq& \left\lfloor \log_p  \frac{p^{p^{e-1}+p^{\frac{e-1}{2}}}}{\sum\limits_{i=0}^{1}\tbinom{p^{e-1}+p^{\frac{e-1}{2}}}{i}(p-1)^i} \right\rfloor\\
  &=& \left\lfloor p^{e-1}+p^{\frac{e-1}{2}}- \log_p \left( 1+p^e+p^{\frac{e+1}{2}}-p^{e-1}-p^{\frac{e-1}{2}} \right) \right\rfloor\\
  &=& p^{e-1}+p^{\frac{e-1}{2}}-e,
\end{eqnarray*}
where the last equality holds as $p^{e-1} < 1+p^e+p^{\frac{e+1}{2}}-p^{e-1}-p^{\frac{e-1}{2}} < p^e$. Thus, $$p^e-e-2=k \leq \mathop{\min}_{t \in \mathbb{Z}^{+}} [rt+k_{opt}^{(q)}(n-t(r+1),d)] \leq p^e-e-1.$$ Then $$\mathop{\min}_{t \in \mathbb{Z}^{+}} [rt+k_{opt}^{(q)}(n-t(r+1),d)]=p^e-e-1 \mbox{ or } p^e-p-2$$ and $\cC_f^{\perp}$ is $k$-optimal or almost $k$-optimal. Similarly, we can also prove the residual cases.
\end{proof}

\begin{remark}
  We remark that there are some optimal or almost optimal locally recoverable codes obtained from $\cC_f$ and its dual, which are verified by Magma. We list them in Table \ref{tab-optimal-loc2}.
  \begin{table}[!h]\footnotesize
\begin{center}
\caption{Optimal or almost optimal locally recoverable codes derived from $\cC_f$ and $\cC_f^{\perp}$.}\label{tab-optimal-loc2}
\begin{tabular}{llll}
\toprule
Conditions & Code & Parameters & Optimality \\
\midrule
$p=3, e=2, f(x)=\tr_{p^e/p}(x^2)$ & $\cC_f$ & $(9,4,4, 3; 2)$ & $k$-optimal and almost $d$-optimal\\
$p=3, e=2, f(x)=\tr_{p^e/p}(x^2)$ & $\cC_f^{\perp}$ & $(9,5,3, 3; 3)$ & $k$-optimal and almost $d$-optimal\\
$p=5, e=2, f(x)=\tr_{p^e/p}(\omega x^2)$ & $\cC_f^{\perp}$ & $(25,21,3, 5; 15)$ & $k$-optimal  and almost $d$-optimal\\
$p=7, e=2, f(x)=\tr_{p^e/p}(x^2)$ & $\cC_f^{\perp}$ & $(49,45,3, 7; 35)$ & $k$-optimal and almost $d$-optimal\\
$p=3, e=3, f(x)=\tr_{p^e/p}(x^2)$ & $\cC_f^{\perp}$ & $(27,22,3,3;14)$ & Almost $k$-optimal\\
$p=3, e=5, f(x)=\tr_{p^e/p}(\omega x^2)$ & $\cC_f^{\perp}$ & $(243,236,3,3;152)$ & Almost $k$-optimal\\
$p=5, e=3, f(x)=\tr_{p^e/p}(x^2)$ & $\cC_f^{\perp}$ & $(125,120,3,5;94)$ & Almost $k$-optimal\\
$p=3, e=4, f(x)=\tr_{p^e/p}(x^{p^{3k}+p^{2k}-p^k+1}+x^2)$ & $\cC_f^{\perp}$ & $(81,75,3,3;50)$ & Almost $k$-optimal\\
\bottomrule
\end{tabular}
\end{center}
\end{table}

\end{remark}

\section{Summary and concluding remarks}\label{sec6}
In this paper, we constructed two families of linear codes $\overline{\cC_D}$ and $\cC_f$ with unbounded lengths from some special functions.
The parameters and weight distributions of them were given under certain conditions.
These two families of codes have the following properties:
\begin{enumerate}
\item[$\diamond$] The locality of $\overline{\cC_D}$ is $2$ for any $N$ and any prime power $q>2$ (see Theorem \ref{loc}).  If $N=q^r+1$, $m=2r$ with $r\geq 3$ and $q=2$, then  $\overline{\cC_D}$ has locality $3$ and $\overline{\cC_D}^{\perp}$ has locality $2^{r-1}(2^{r-1}-1)-1$. In particular, $\overline{\cC_D}^{\perp}$ is a $k$-optimal or almost $k$-optimal LRC (see Theorem \ref{th-loc-q=2}). Besides, $\cC_f$ has locality $3$ and $\cC_f^{\perp}$ has locality $p^2-p-2$ for some cases (see Theorem \ref{loc-cf}). In these cases, $\cC_f$ is a $k$-optimal LRC and $\cC_f^{\perp}$ is a $k$-optimal and $d$-optimal LRC. In other cases, we conjectured that $\cC_f$ has locality $2$ and $\cC_f^{\perp}$ has locality $p^e-p^{e-1}-p^{\frac{e-1}{2}}-1$, $p^e-p^{e-1}-(p-1)p^{\frac{e-2}{2}}-1$, or $p^e-p^{e-1}-p^{\frac{e-2}{2}}-1$ for three different cases (see Conjecture \ref{conj}). If this conjecture can be tackled, then we can derive some infinite families of almost optimal or optimal locally recoverable codes from $\cC_f^{\perp}$ (see Corollary \ref{coro}).
\item[$\diamond$] The code  $\overline{\cC_D}$ was proved to be  almost optimally extendable or optimally extendable if $N=q^r+1$ for $m=2r$ and $N=2$ for odd $m$ (see Theorems \ref{extend1} and \ref{extend2}).
The code $\cC_f$ was proved to be  optimally extendable for the most cases and almost optimally extendable for a few cases (see Theorem \ref{th-extend}).
\item[$\diamond$] Both of the code $\overline{\cC_D}$ for $N=q^r+1$ with $m=2r$ or $N=2$ with odd $m$ and the code $\cC_f$ are self-orthogonal (see Theorems \ref{tem-so1}, \ref{tem-wtCDN2} and \ref{tem-wtCD2}).
\item[$\diamond$] Some families of linear codes holding $2$-designs were also derived (see Corollaries \ref{col-design1} and \ref{col-design2}).
\end{enumerate}
Hence, the codes in this paper have nice applications in distributed storage systems,  the implementations of block ciphers, combinatorics and many other fields. From the weight distribution of the linear codes constructed in this paper, we find that the minimum distance of $\overline{\cC_D}$ is the same with that of $\cC_{D\setminus \{0\}}$ studied in literatures \cite{HZ4, D4}. However, the dimension and the code rate of $\overline{\cC_D}$ are both larger than those of $\cC_{D\setminus \{0\}}$. We also find that the minimum distance of $\cC_f$ is the same with that of $\cC_f^*$ studied in literature \cite{MS}.
However, the dimension and the code rate of $\overline{\cC_D}$ are both larger than those of $\cC_f^*$.  As was pointed out by Ding and Tang in \cite{D}, it is usually a challenge to determine the parameters of the augmented code of a linear code as we may require the complete weight information of this linear code.  Furthermore,
in Table \ref{tab-optimalcode}, we list many (almost) optimal codes derived in this paper.

\begin{table}[!h]
\footnotesize
\begin{center}
\caption{Optimal codes or almost optimal codes derived in Theorem \ref{tem-wtCD1}, \ref{tem-wtCD2}, \ref{th-LCD1}, \ref{th-LCD2} and \ref{th-LCD3}.}\label{tab-optimalcode}
\begin{tabular}{llll}
\toprule
Conditions & Code & Parameters & Optimality \\
\midrule
$q=2,r=3$ & $\overline{\cC_D}(N=q^r+1)$ & $[28, 7, 12]$ & Optimal\\
$q=2,r=3$ & $\overline{\cC_D}^{\perp}(N=q^r+1)$ & $[28,21,4]$ & Optimal\\
$q=2,r=4$ & $\overline{\cC_D}(N=q^r+1)$ & $[120,9,56]$ & Optimal\\
$q=2,r=4$ & $\overline{\cC_D}^{\perp}(N=q^r+1)$ & $[120, 111, 4]$ & Optimal\\
$q=3,r=2$ & $\overline{\cC_D}(N=q^r+1)$ & $[21, 5, 12]$ & Optimal\\
$q=3,r=2$ & $\overline{\cC_D}^{\perp}(N=q^r+1)$ & $[21, 16, 3]$ & Optimal\\
$q=3,r=2$ & $\overline{\cC_D}'^{\perp}(N=q^r+1)$ & $[26, 21, 3]$ & Optimal\\
$q=3,r=3$ & $\overline{\cC_D}(N=q^r+1)$ & $[225,7,144]$ & Optimal\\
$q=3,r=3$ & $\overline{\cC_D}^{\perp}(N=q^r+1)$ & $[225,218,3]$ & Optimal\\
$q=3,r=3$ & $\overline{\cC_D}'^{\perp}(N=q^r+1)$ & $[232,225,3]$ & Optimal\\
$q=4,r=2$ & $\overline{\cC_D}(N=q^r+1)$ & $[52, 5, 36]$ & Optimal\\
$q=4,r=2$ & $\overline{\cC_D}^{\perp}(N=q^r+1)$ & $[52, 47, 3]$ & Optimal\\
$q=4,r=2$ & $\overline{\cC_D}'^{\perp}(N=q^r+1)$ & $[57, 52, 3]$ & Optimal\\
$q=5,r=2$ & $\overline{\cC_D}(N=q^r+1)$ & $[105, 5 ,80]$ & Optimal\\
$q=5,r=2$ & $\overline{\cC_D}^{\perp}(N=q^r+1)$ & $[105, 100, 3]$ & Optimal\\
$q=5,r=2$ & $\overline{\cC_D}'^{\perp}(N=q^r+1)$ & $[110, 105, 3]$ & Optimal\\
$q=3,m=5$ & $\overline{\cC_D}^{\perp}(N=2)$ & $[81, 75, 3]$ & Optimal\\
$q=3,m=5$ & $\overline{\cC_D}'^{\perp}(N=2)$ & $[87, 81, 3]$ & Optimal\\
$q=5,m=3$ & $\overline{\cC_D}^{\perp}(N=2)$ & $[25, 21, 3]$ & Almost optimal\\
$q=5,m=3$ & $\overline{\cC_D}'^{\perp}(N=2)$ & $[29, 25, 3]$ & Optimal\\
$q=7,m=3$ & $\overline{\cC_D}^{\perp}(N=2)$ & $[49, 45, 3]$ & Almost optimal\\
$q=7,m=3$ & $\overline{\cC_D}'^{\perp}(N=2)$ & $[53, 49, 3]$ & Optimal\\
$q=9,m=3$ & $\overline{\cC_D}^{\perp}(N=2)$ & $[81, 77, 3]$ & Almost optimal\\
$q=9,m=3$ & $\overline{\cC_D}'^{\perp}(N=2)$ & $[85, 81, 3]$ & Optimal\\
$p=3,e=3$& $\cC_f$ & $[27,5,15]$ & Almost optimal\\
$p=3,e=3$& $\cC_f^{\perp}$ & $[27,22,3]$ &  Optimal\\
$p=3,e=3$& $\cC_f'^{\perp}$ & $[32,27,3]$ &  Optimal\\
$p=3,e=4$& $\cC_f$ & $[81,6,51]$ & Optimal\\
$p=3,e=4$& $\cC_f^{\perp}$ & $[81,75,3]$ &  Optimal\\
$p=3,e=4$& $\cC_f'^{\perp}$ & $[87,81,3]$ &  Optimal\\
$p=5,e=2$& $\cC_f$ & $[25,4,19]$ &  Optimal\\
$p=5,e=2$& $\cC_f'$ & $[29,4,20]$ & Almost optimal\\
$p=5,e=2$& $\cC_f^{\perp}$ & $[25,21,4]$ &  Optimal\\
$p=5,e=2$& $\cC_f'^{\perp}$ & $[29,25,3]$ &  Optimal\\
$p=5,e=3$ & $\cC_f^{\perp}$ & $[125, 120, 3]$ &  Optimal\\
$p=5,e=3$ & $\cC_f'^{\perp}$ & $[130, 125, 3]$ &  Optimal\\
$p=7,e=2$ & $\cC_f$ & $[49, 4, 41]$ &  Optimal\\
$p=7,e=2$& $\cC_f'$ & $[53,4,42]$ & Almost optimal\\
$p=7,e=2$& $\cC_f^{\perp}$ & $[49, 45, 4]$ &  Optimal\\
$p=7,e=2$& $\cC_f'^{\perp}$ & $[53, 49, 3]$ &  Optimal\\
\bottomrule
\end{tabular}
\end{center}
\end{table}

\end{document}